%% file: Convex-FPA.tex
\pdfoutput=1

\documentclass[11pt]{article}
\usepackage{amsfonts}
\usepackage{amssymb,amsmath,amsthm}
\usepackage{graphicx}
\usepackage[dvips, paper=letterpaper, top=1in, bottom=1in, left=1in, right=1in, nohead, includefoot, footskip=0.4in]{geometry}
\usepackage[authoryear,round]{natbib}
\usepackage{environ}
\usepackage{color}
\usepackage{epsfig}
\usepackage{caption}
\usepackage{verbatim}
\usepackage{subcaption}
\usepackage{float}
\usepackage{setspace}
\usepackage{enumitem}
\setlist[enumerate]{itemsep=0mm}
\usepackage{calc}
\usepackage[pdfpagemode=UseOutlines,hyperfootnotes=false, plainpages=false]{hyperref}
\usepackage{mathtools}
\usepackage{dsfont}
\usepackage[english]{babel}
\usepackage{amsthm}
\usepackage{bm}
\usepackage{breqn}
\usepackage{bigints}
\usepackage{tikz}
\usetikzlibrary{positioning}
\usepackage{nicefrac}
\usepackage{xurl}
\usepackage[linesnumbered,ruled,vlined]{algorithm2e}
\usepackage{algpseudocode}
\usepackage{subfiles}
\usepackage[capitalise]{cleveref}
\usepackage[flushleft]{threeparttable}
\usepackage{array,booktabs,makecell}
\usepackage{subcaption}
\usepackage{booktabs}

\usepackage[font=normalsize,labelfont=bf]{caption}
\usepackage[font=footnotesize,labelfont=bf]{subcaption}
\usepackage{sectsty}
\allsectionsfont{\normalfont\sffamily\bfseries}
\sffamily

\newtheoremstyle{theoremsansserif} 
    {\topsep}                    
    {\topsep}                    
    {\itshape}                   
    {}                           
    {\sffamily\bfseries }        
    {.}                          
    {.5em}                       
    {}  

\theoremstyle{theoremsansserif}

\newtheorem{lemma}{Lemma}
\newtheorem{corollary}{Corollary}
\newtheorem{proposition}{Proposition}
\newtheorem{remark}{Remark}

\newtheorem*{example*}{Example}

\newtheorem{theorem}{Theorem}

\theoremstyle{definition}

\newcommand{\R}{\mathbb{R}}

\newcommand{\E}{\mathbb{E}}
\newcommand{\p}{\pmb{p}}

\newcommand{\e}{\pmb e}
\renewcommand{\d}{\pmb d}
\renewcommand{\v}{\pmb v}
\newcommand{\vv}{\mathcal{V}}
\newcommand{\pp}{\mathcal{P}}
\newcommand{\grad}{\nabla}
\DeclareMathOperator{\rev}{Rev}
\DeclareMathOperator{\argmin}{argmin}
\DeclareMathOperator{\argmax}{argmax}
\DeclareMathOperator{\regret}{Regret}
\DeclareMathOperator{\mye}{Myer}


\begin{document}
\title{\sf\textbf{Strategically-Robust Learning Algorithms for Bidding in First-Price Auctions}}
\author{
\sf Rachitesh Kumar\\
\sf Columbia University \\
\small\texttt{rk3068@columbia.edu}
\and
\sf Jon Schneider \\
\sf Google Research\\
\small\texttt{jschnei@google.com}
\and
\sf Balasubramanian Sivan\\
\sf Google Research \\
\small\texttt{balusivan@google.com}}
\date{\today}

\maketitle

\allowdisplaybreaks

\begin{abstract}
    \normalsize
    \input{new_abstract.tex}
\end{abstract}
\newpage

\setcounter{page}{1}

\setstretch{1.5}

\graphicspath{ {Images/} }

\input{intro.tex}

\input{model.tex}

\input{known_distribution}

\input{unknown_dist}

\input{log_regret}

\section{Acknowledgements}

The authors would like to thank Santiago Balseiro and Christian Kroer for helpful discussions, and providing valuable feedback on a preliminary version of this paper.

\singlespacing

\bibliographystyle{plainnat}
\bibliography{refs}

\setstretch{1.5}

\appendix

\newpage
\pagenumbering{arabic}\renewcommand{\thepage}{ec \arabic{page}}

\begin{centering}
\LARGE
Electronic Companion:\\[1em]
Strategically-Robust Bidding Algorithms for First-Price Auctions\\[1em]
\large
Rachitesh Kumar, Jon Schneider, Balasubramanian Sivan\\[1em]
\today\\
\end{centering}
\renewcommand{\theequation}{\thesection-\arabic{equation}
}

\input{appendix_concave.tex}
\input{appendix_known.tex}
\input{appendix_unknown_dist}

\end{document}

%% file: new_abstract.tex
Learning to bid in repeated first-price auctions is a fundamental problem at the interface of game theory and machine learning, which has seen a recent surge in interest due to the transition of display advertising to first-price auctions. In this work, we propose a novel concave formulation for pure-strategy bidding in first-price auctions, and use it to analyze natural Gradient-Ascent-based algorithms for this problem. Importantly, our analysis goes beyond regret, which was the typical focus of past work, and also accounts for the strategic backdrop of online-advertising markets where bidding algorithms are deployed---we provide the first guarantees of strategic-robustness and incentive-compatibility for Gradient Ascent.

Concretely, we show that our algorithms achieve $O(\sqrt{T})$ regret when the highest competing bids are generated adversarially, and show that no online algorithm can do better. We further prove that the regret reduces to $O(\log T)$ when the competition is stationary and stochastic, which drastically improves upon the previous best of $O(\sqrt{T})$. Moving beyond regret, we show that a strategic seller cannot exploit our algorithms to extract more revenue on average than is possible under the optimal mechanism. Finally, we prove that our algorithm is also incentive compatible---it is a (nearly) dominant strategy for the buyer to report her values truthfully to the algorithm as a whole. Altogether, these guarantees make our algorithms the first to simultaneously achieve both optimal regret and strategic-robustness. 

%% file: intro.tex
\section{Introduction}

Advertising is an indispensable part of the internet economy. It allows online platforms (like Google and Meta) to monetize their services by charging advertisers for the opportunity to display their ads to users. This is operationally achieved through an online market/exchange where advertising opportunity is sold to interested advertisers. The mechanism of choice for these online advertising markets is real-time auctions: anytime a user visits the platform, an auction is run to determine the advertiser who will get to display their ad to that user, and the payment to be charged for that opportunity. Each of these auctions runs in less than a few milliseconds and advertisers typically participate in millions of such auctions as part of their advertising campaign. Given the speed and scale of ad-auction markets, all bidding is done programmatically---each advertiser employs an automated bidding algorithm, which is often provided as a service by the platform itself or a third-party demand-side platform (DSP). This algorithm takes as input high-level objectives like value-per-click, targeting criteria etc., and bids on behalf of the advertiser with the goal of maximizing her utility.

Until a few years ago, the second-price auction and its generalizations were the dominant auction formats in online advertising, but that is no longer the case with the transition of the display-advertising industry to first-price auctions~\citep{Google}. Unlike second-price auctions, where truthful bidding is optimal, bidding in first-price auctions is highly non-trivial and presents the need for non-trivial bidding algorithms. Combined with the colossal scale of online advertising markets and the accompanying abundance of data, this transition has created a need for online algorithms for bidding in repeated first-price auctions that can learn from data. Thus motivated, a recent line of work \citep{balseiro2022contextual, han2020optimal, han2020learning, zhang2022leveraging, wang2023learning, badanidiyuru2023learning} has proposed algorithms for a variety of input models (adversarial, stochastic, etc.) and feedback structures (bandit, partial, full, etc.). These works analyze the problem of bidding in repeated first-price auctions through the lens of online learning, and consequently focus on minimizing regret against the best fixed bidding strategy in hindsight. However, although regret is an important aspect of any learning algorithm, it completely ignores the strategic nature of the markets in which these algorithms are deployed---both the buyer (advertiser) and the seller (platform) can attempt to manipulate the algorithm in order to obtain better revenue/utility. Consequently, bidding algorithms with strong regret guarantees can perform very poorly when deployed in markets with strategic agents. \citet{braverman2018selling} showed that this is true of all mean-based algorithms, which includes bidding algorithms based on popular paradigms like Exponential Weights/Hedge~\citep{freund1997decision}, EXP3~\citep{auer2002nonstochastic}, Follow-the-Regularized-Leader (FTRL), Sample-Average-Approximation etc., and importantly includes nearly all algorithms proposed in recent works (see Subsection~\ref{subsec:related}). In particular, \citet{braverman2018selling} showed that any mean-based algorithm is susceptible to manipulation on either side of the market:
\begin{itemize}[topsep=0pt]
    \item A seller who knows that the buyer is employing a mean-based algorithm can extract more revenue on average than is possible under the optimal single-shot mechanism (posting the monopoly reserve price). Moreover, she can do so by simply posting a sequence of decreasing reserve prices.
    \item A buyer can improve her utility by misreporting her values to any automated bidding algorithm that is mean based.
\end{itemize}
\citet{braverman2018selling} go on to propose an algorithm that attains sub-linear regret while being resistant to manipulation by the seller and incentivizing truthful reporting by the buyer. It does so by minimizing a more complicated notion of regret---namely contextual-swap regret---instead of standard (external) regret. Although this approach is strategically robust, it is complicated and comes at a substantial cost: (i) the algorithm of \citet{braverman2018selling} suffers from $\Theta(T^{7/8})$-regret, which is much worse than the $O(\sqrt{T})$-regret achieved by other works (e.g., \citealt{balseiro2022contextual}); (ii) it requires a super-constant $\Omega(T^{1/8})$ amount of computation and memory for each auction.\footnote{\citet{braverman2018selling} study a model with $m$ possible values and provide guarantees in terms of $m$. Here we use $m = T^{1/8}$ as the discretization because it optimally trades off regret and strategic robustness in their guarantees.} This naturally begs the question:
\begin{quote}
\vspace{-0.5em}
\emph{Is it possible to design an algorithm that is both strategically robust and achieves the optimal $O(\sqrt{T})$-regret rate? Can it be done with a simple algorithm that only requires a constant amount of computation and memory for each auction?}
\end{quote}
In light of the vulnerability of all natural mean-based algorithms to strategic manipulation and the universal reliance on mean-based algorithms to achieve $O(\sqrt{T})$-regret in prior work, one might be tempted to conclude that the aforementioned questions do not admit a positive answer. However, such a conclusion would ignore perhaps the most important/popular data-driven optimization algorithm in existence---(Online) Gradient Descent/Ascent, which is not a mean-based algorithm. In fact, the literature on bidding in first-price auctions, including \citet{braverman2018selling}, is marked by a conspicuous lack of results for Gradient Ascent, and leaves open the fundamental question:
\begin{quote}
\vspace{-0.5em}
\emph{Is Online Gradient Ascent strategically robust as a bidding algorithm?}
\end{quote}
Surprisingly, we show that the answer to all of the above questions is a resounding `Yes!'. We provide a comprehensive analysis of Online Gradient Ascent for first-price auctions, and show that this simple algorithm achieves the optimal regret rate (with a constant amount of compute and memory) while also being robust to strategic manipulations by both the seller and the buyer; making it the first to achieve this hitherto-unattained amalgam of guarantees.

\subsection{Major Contributions}

We study the problem of designing algorithms for a buyer who participates in $T$ sequential first-price auctions with a continuum of values and discrete bids. We assume that the value of the buyer is drawn independently from some distribution $F$ (with bounded density) in each auction and that she observes this value before bidding. Moreover, motivated by practice, we assume full feedback---the maximum of the highest competing bid and the reserve price is revealed after each auction. The performance of algorithms is measured using regret against the best fixed bidding strategy (map from values to bids) in hindsight. Below, we provide a brief overview of our results.

\textbf{Concave Formulation.} The bedrock of our algorithms and analysis is a novel concave formulation for \emph{pure-strategy} bidding in first-price auctions (Theorem~\ref{thm:convex-refor}). In particular, we propose a change-of-variables transformation that maps each (monotonic) pure bidding strategy to the probability distribution over bids induced by that strategy and the randomness in the buyer's value. Importantly, while utility is not concave as a function of the bidding strategy (which itself is an infinite-dimensional object), we show that it is concave as a function of the induced probability distribution over bids. As the set of possible bids is finite (e.g. discretized to cents), we are able to transform the problem of finding the optimal bidding strategy in first-price auctions from an infinite dimensional non-concave problem to a finite-dimensional concave one. Notably, this transformation does not impose any regularity conditions on the value distribution $F$, as is often the case with such transformations in other contexts (like \citealt{bulow1989simple} and \citealt{kinnear2022convexity}). To the best of our knowledge, this is the first unconditional concave formulation for pure-strategy bidding in first-price auctions, and may be of independent interest. 
In this work, we leverage it to propose two algorithms: Algorithm~\ref{alg:known_GA} requires knowledge of the value distribution $F$ and simply implements Online Gradient Ascent for the online concave maximization problem implied by our formulation; Algorithm~\ref{alg:threshold} does not require knowledge of the value distribution $F$ and instead implements Online Gradient Ascent under the pretense of uniformly-distributed values. 

\textbf{Regret Guarantees against Adaptive Adversarial Inputs and Stochastic Inputs.} Algorithm~\ref{alg:known_GA} inherits the $O(\sqrt{T})$-regret guarantee enjoyed by Online Gradient Ascent under adversarial input (Proposition~\ref{prop:sub-linear-regret}). We show that Algorithm~\ref{alg:threshold} also enjoys the same $O(\sqrt{T})$-regret guarantee against adversarial highest competing bids (Theorem~\ref{thm:regret_threshold}), despite falling outside the purview of Online Convex Optimization and not directly inheriting any properties of Online Gradient Ascent. Our regret guarantees hold even if the highest competing bids are chosen adaptively based on the past, and thus apply to the setting where all of the buyers are simultaneously learning to bid. Moreover, we show that our guarantees are tight---no algorithm can achieve $o(\sqrt{T})$-regret against adversarial competition (Proposition~\ref{prop:regret-lower-bound}). When the competition is stochastic, i.e., the maximum of the highest competing bid and the reserve price is i.i.d. from some distribution, our concave reformulation yields a strongly-concave optimization problem. This allows us to prove a $O(\log T)$-regret guarantee for Algorithm~\ref{alg:known_GA} (Theorem~\ref{thm:log_regret}), which exponentially improves over the previous-best $O(\sqrt{T})$-regret bound~\citep{balseiro2022contextual,han2020optimal}.

\textbf{Strategic Robustness against the Seller.} If a buyer employs either of our algorithms for bidding, we prove that the seller cannot extract more revenue than $\mye(F)\cdot T + O(\sqrt{T})$ from her in total, where $\mye(F)$ is the optimal revenue obtainable under any single-shot mechanism for value distribution $F$, i.e, $\mye(F)$ is the revenue obtained by posting the monopoly reserve price. Put another way, the seller cannot exploit our algorithms to extract (substantially) more revenue than $\mye(F)$ on average (Theorem~\ref{thm:strat_robust} and Theorem~\ref{thm:threshold_robust}). In particular, the seller does not gain by dynamically changing the reserve price and is best off just posting the monopoly reserve price in each auction. This robustness to strategic manipulation is critical in practice because platforms often have detailed knowledge of the bidding algorithms of the advertisers (they might even design them!). Our algorithms also lead to more stability by removing the incentive for the seller to manipulate the bidding algorithm through dynamic reserve prices. 
Finally, Algorithm~\ref{alg:threshold} is also strategically robust in the multi-buyer setting where all of the buyers simultaneously use it to bid. In that setting, it limits the seller's revenue to $\mye(\{F_i\}_i) + O(\sqrt{T})$, where $\mye(\{F_i\}_i)$ is the revenue of the optimal mechanism for buyers with $\prod_i F_i$ as the prior. Once again, this implies that the seller cannot extract more average revenue than is possible under the optimal single-shot mechanism (Theorem~\ref{thm:multi_buyer_robust}).

\textbf{Incentive Compatibility for the Buyer.} Bidding algorithms are deployed as automated agents that bid on behalf of the buyers (advertisers) in the auctions. The buyers provide high-level information about their value for winning these auctions by specifying their value per click/conversion/impression and targeting criteria, which is then used by the automated bidding algorithm to optimize utility. As the high-level information is private to the buyer, one cannot simply assume that she will reveal it truthfully to the algorithm. If the algorithm incentivizes misreporting of values by the buyer, it hurts the seller because she does not obtain reliable data about the values of the buyer (which is very valuable for experimentation), and it hurts the buyer because it imposes the burden of computing the best misreporting strategy. Algorithm~\ref{alg:threshold} does not suffer from these shortcomings and incentivizes truthful reporting of values by the buyer. Formally, we show that the excess utility over truthful reporting earned through any misreport of values is no more than $O(\sqrt{T})$ (Theorem~\ref{thm:ic_threshold}). In other words, the \emph{online algorithm itself} is (approximately) incentive compatible with respect to its input. This property holds for adaptively adversarial highest competing bids. In particular, if all buyers are simultaneously learning to bid, then the buyers using Algorithm~\ref{alg:threshold} would not prefer to misreport their values in hindsight.

\subsection{Significance}\label{subsec:sig}
We highlight a few significant implications of our work here.

\textbf{Agile vs Lazy Projections Makes All the Difference.}
Given the close connections between our algorithm (Online Gradient Ascent) and FTRL with the Euclidean regularizer (which is mean-based)--- namely that they differ only in how they handle projections, with the former being ``agile'' and the latter being ``lazy'' (see \citealt{hazan2016introduction} for a discussion), one might have little hope for Online Gradient Ascent. I.e., one might expect Online Gradient Ascent to also suffer from strategic manipulability, much like FTRL with Euclidean regularizer. Surprisingly, our analysis reveals that the shift from lazy to agile projections makes a substantial difference in the robustness properties of the algorithm. Intuitively, the fact that lazy projections result in way too much memory of the past, and hence makes the algorithm mean-based and consequently strategically manipulable is not hard to digest. However the fact that this simple switch from lazy to agile projections makes the algorithm strategically robust against both the buyer and the seller (recall that by the result of~\citet{braverman2018selling}, both these properties are false for lazy projections), while obtaining the $O(\sqrt{T})$ regret guarantee even against adaptively generated adversarial inputs is quite surprising.

\textbf{Swap Regret Minimization.}
Interestingly, all previously-known algorithms achieving this sense of strategic robustness required some form of swap-regret minimization, with \citep{MMSSbayesian, RZ24} even showing that low swap regret is a necessary property for a learning algorithm to be generically strategically-robust across all games. Our work is the first to achieve these properties without explicitly minimizing some form of swap regret. This is significant because, unlike Follow-The-Regularized-Leader (FTRL) and Gradient Ascent, swap-regret-minimizing algorithms are often non-intuitive and complex. Rather than designing swap-regret minimizing algorithms to overcome the strategic weaknesses of FTRL, our results imply that one can instead simply make a small switch from lazy to agile projections. This makes the algorithm retain all its intuitiveness, and as we prove, brings in all the desired robustness. Given the results of~\citep{MMSSbayesian, RZ24}, one might wonder whether it is possible to establish strategic robustness of an algorithm without explicitly establishing a low value for the relevant notion of swap regret for that algorithm. However, the necessity of using algorithms with low swap-regret values established in those two works (the former for non-Bayesian games, and the latter for Bayesian games) only holds if we seek strategic robustness in all (respectively non-Bayesian or Bayesian) games. Our work is focused on the specific game of first-price auctions, and shows that we are not bound by those results---we directly establish strategic robustness without first establishing a low value for some notion of swap regret.

\subsection{Additional Related Work}\label{subsec:related}

The first-price auction is arguably the most popular auction format in human history. It has been studied extensively in the economics literature, where the primary focus has been on the analysis of equilibria. Since our focus is on developing data-driven bidding algorithms, we do not discuss the work on equilibrium analysis here, and refer the reader to standard texts on auction theory like \citet{krishna2009auction} and \citet{milgrom2004putting}. Similarly, we omit the work on equilibrium analysis in computer science and operations research, and refer to recent works of \citet{chen2023complexity} and \citet{balseiro2023contextual} for an overview. Finally, extensive work has been done at the intersection of auctions and data-driven optimization, the vast majority of which has focused on the mechanism-design problem faced by the seller. We refer to the survey by \citet{nedelec2022learning} for a detailed discussion of learning algorithms for buyers and sellers in repeated auctions, and focus here exclusively on relevant work on bidding algorithms for first-price auctions and strategic aspects of learning.

Motivated by the change of auction format in the Display Advertising industry, \citet{balseiro2022contextual} analyze bidding algorithms for first-price auctions. They assume that the buyer only observes binary feedback, i.e., whether or not she won the auction, and model it as a contextual bandit problem with potential cross-learning between contexts. When the highest competing bids are stochastic, they propose a UCB-based algorithm which achieves $O(\sqrt{T})$ regret, whereas when the highest competing bids are adversarially generated, they propose an EXP3-based algorithm which achieves $O(T^{2/3})$ regret (this was later improved to $O(\sqrt{T})$ regret by \cite{schneider2024optimal}).  \citet{han2020optimal} also study a model where the highest competing bids are stochastic, albeit under a different partial-feedback model where only the winning bid is revealed at the end of each auction. They propose algorithms that achieve $O(\sqrt{T})$ regret and allow for infinitely many possible bids. \citet{han2020learning} study a model where both the values and the highest competing bids are adversarial, but restrict the space of benchmark strategies to be Lipschitz. They propose an algorithm that runs the Exponential Weights Algorithm over a suitable cover of the space of all Lipschitz bidding strategies, and prove a $O(\sqrt{T})$-regret guarantee for it. Moreover, they show that a computationally-expensive version of their algorithm attains $O(\sqrt{T})$-regret when the benchmark strategies are monotonic instead of Lipschitz. \citet{zhang2022leveraging} extend the algorithm and analysis of \citet{han2020learning} to incorporate hints about the highest competing bids. 

\citet{badanidiyuru2023learning} study a contextual model of first-price auctions in which the highest competing bids are the sum of a linear function of the contexts and a log-concave stochastic noise term. They propose algorithms that attain $O(\sqrt{T})$ regret under different feedback and informational assumptions. \citet{wang2023learning} analyze repeated first-price auctions with a global budget constraint. When the highest competing bids are stochastically generated, they propose dual-based algorithms which achieve $O(\sqrt{T})$ regret under both full and partial feedback. Importantly, all of these aforementioned algorithms are mean based, and consequently they are neither strategically robust nor incentive compatible \citep{braverman2018selling}. \citet{feng2018learning} and \citet{cesa2023role} investigate bidding in repeated auctions when the buyer does not know her own value, and propose algorithms that compete against the best static bid in hindsight. We assume that the buyer observes her value before bidding and use the best \emph{strategy} in hindsight as the benchmark, and therefore our results are not directly comparable. \citet{kinnear2022convexity} study the different but related problem of procuring advertising opportunity for contract fulfilment. They analyze the full-information optimization problem against stochastic competition, and reformulate it as a convex optimization problem in the space of winning probabilities. Unlike our unconditional concave formulation, their convex formulation for first-price auctions requires the competing bid distribution to have full support and satisfy a log-concavity-like assumption.

Finally, our paper is closely connected to a growing body of literature on strategizing against no-regret learning algorithms in games. This area of work is concerned with the two questions of: 1. how should you best-respond if you know other players in a repeated game choose their actions according to a learning algorithm? and 2. what learning algorithm should you choose to be robust to the strategic behavior of other players? \cite{braverman2018selling} was one of the first works to investigate these questions, specifically for the setting of non-truthful auctions -- their work was later generalized to the prior-free setting \citep{deng2019prior}, the setting of multiple buyers \citep{cai2023selling}, and the setting of selecting parameters for bidding algorithms \citep{kolumbus2022and, kolumbus2022auctions}. Since then, these questions have also been studied in the settings of general games \citep{DSSstrat, brown2023is}, Bayesian games \citep{MMSSbayesian}, contract design \citep{guruganesh2024contracting} and Bayesian persuasion \citep{chen2023persuading}.

%% file: model.tex
\section{Model}

\paragraph{Notation.} $\R_+ \coloneqq [0,\infty)$ denotes the set of non-negative reals. We use boldface for vectors. If a vector is indexed by time, like $\pmb a_t$, then its $i$-th component is denoted by $a_{t,i}$. Throughout the paper, $\|\cdot\|$ represents the Euclidean norm, i.e., $\|\pmb a\| = (\sum_{i} a_i^2)^{1/2}$.

Consider a buyer who participates in $T$ sequential first-price auctions. In each auction $t \in [T]$, her value $V_t \in [0,1]$ for the item is drawn independently from a distribution with CDF $F: [0,1] \to [0,1]$, and bounded density $f: [0,1] \to \R_+$ such that $\sup_{x \in [0,1]} f(v) \leq \bar f$. In line with practice, we will assume that the set of possible bids is finite and equally spaced (e.g., multiples of cents): there are $K+1$ possible bids $0 = b_0 < b_1 < \dots < b_K \leq 1$, where $b_i = i \cdot \epsilon$ for some $0 < \epsilon \leq 1/K$. We use $ h_t \in \{b_0, b_1, \dots, b_K\}$ to denote the minimum bid needed to win at time $t$, i.e., $h_t$ is the maximum of the highest competing bid and the reserve price (see Subsection~\ref{sec:tie-breaking} for a detailed discussion, including the impact of tie-breaking). To simplify terminology, we will treat the reserve price as an additional bid submitted by the seller and often refer to $h_t$ simply as the highest competing bid.

In auction $t \in [T]$, the following sequence of events takes place:
\begin{itemize}
    \item Nature picks the highest competing bid $h_t$. Aside from Section~\ref{sec:log_regret}, we allow nature's choice to be adaptively adversarial, i.e., $h_t$ can be chosen arbitrarily based on the past $\{s_r(\cdot), V_r\}_{r=1}^{t-1}$.
    \item The buyer observes her value $V_t \sim F$ and places a bid $s_t(V_t) \in \{b_0, b_1, \dots, b_K\}$.
    \item The buyer wins the item and pays $s_t(V_t)$ if $s_t(V_t) \geq h_t$. If $s_t(V_t) < h_t$, she does not win the item and does not make any payment.
    \item The buyer observes the highest competing bid $h_t$.\footnote{Many platforms reveal the minimum bid needed to win in practice to help advertisers bid more efficiently, e.g., see \url{https://support.google.com/authorizedbuyers/answer/12798257?hl=en}.}
\end{itemize}

As the value $V_t$ can lie anywhere in $[0,1]$, the buyer effectively specifies a bidding strategy $s_t: [0,1] \to \{b_0, b_1, \dots, b_K\}$ at each time $t$, where $s_t(V_t)$ is the bid in auction $t$ if her value is $V_t$. When the buyer employs the strategy $s: [0,1] \to \{b_0, b_1, \dots, b_K\}$ and the highest competing bid is $h \in \{b_0, b_1, \dots, b_K\}$, her expected utility is given by
\begin{align*}
	u(s | F, h) \coloneqq \E_{v \sim F} \left[ (v - s(v)) \cdot \mathbf{1}(s(v) \geq h) \right] \,.
\end{align*}

An online bidding algorithm $A$ for the buyer is a (potentially randomized) procedure which specifies a bidding strategy $A_t:[0,1] \to \{b_0, b_1, \dots, b_K\}$ at each time $t$, based only on the information observed till time $t-1$ and the value $V_t$. We will measure the learning rate of an online algorithm by its (pseudo) regret compared to the best static bidding strategy $s$:
\begin{align*}
    \regret(A| F) \coloneqq \max_{s(\cdot)} \sum_{t=1}^T \E[u(s| F, h_t)]\ -\ \sum_{t=1}^T \E[u(A_t| F, h_t)] \,, 
\end{align*}
where the expectation is over any randomness in $h_t$ (which can potentially depend on the randomness in historical values $\{V_r\}_{r=1}^{t-1}$). All our algorithms will output strategies $A_t(\cdot)$ that deterministically depend on the historical highest competing bids $\{h_s\}_{s=1}^{t-1}$. Therefore, the adversary can compute $A_t(\cdot)$ using the past $\{A_r(\cdot), V_r\}_{r=1}^{t-1}$ and select $h_t$ based on it. In other words, we allow nature to choose $h_t$ based on $A_t(\cdot)$, in addition to the past $\{s_r(\cdot), V_r\}_{r=1}^{t-1}$. However, it is worth noting that $h_t$ cannot depend on the private value $V_t$ of the buyer.

\subsection{Multiple Buyers and Tie Breaking}\label{sec:tie-breaking}

Our focus in this paper is on developing algorithms for individual buyers. Consequently, the bids of competing buyers and the reserve price of the seller are only relevant in so far as they determine the bids at which the buyer under consideration wins the auction. Here, we argue that the minimum bid needed to win, denoted by $h_t$, completely captures the effect of competing bids and the reserve price. To see this, suppose there are $n-1$ competing buyers in auction $t$. Let $\beta(1)_t, \beta(2)_t, \dots, \beta(n-1)_t$ be their (potentially random) bids, and $r_t$ be the reserve price. Moreover, let $\bar \beta_t = \max_{j} \beta(j)_t$ denote the highest bid among the competing buyers, and $\Gamma \coloneqq \{i \in [n-1]\mid \beta(i)_t = \bar \beta_t\}$ be the set of competing buyers which are tied for it.

First observe that any natural tie-breaking rule, including uniform and lexicographic, can be implemented using random rankings: draw a permutation of buyers $\sigma_t \in S_n$ independently according to some distribution and break ties in favor of higher ranked buyers. For example, one can implement the uniform tie-breaking by picking a permutation uniformly at random from $S_n$. Now, conditioned on the random ranking $\sigma_t$, the buyer under consideration wins auction $t$ if and only if her bid clears the reserve price, i.e., $s_t(V_t) \geq r_t$, \emph{and} one of the following conditions is satisfied:
\begin{enumerate}
    \item $s_t(V_t) = \bar \beta_t$ and her rank is higher than all competing buyers in $\Gamma$
    \item $s_t(V_t) > \bar \beta_t$.
\end{enumerate}
If condition (1) is satisfied, we set $h_t = \max\{r_t, \bar \beta_t\}$, and if condition (2) is satisfied, we set $h_t = \max\{r_t, \bar \beta_t + \epsilon\}$. Finally, it is possible that $\bar \beta_t = b_K$ (i.e., $\bar \beta_t$ is the highest bid possible) and some competing buyer in $\Gamma$ who made that bid outranks the buyer under consideration. In which case, it is impossible for the buyer under consideration to win the auction and her utility is identically 0 for all bids. We can safely ignore such auctions without loss of generality, and we assume it is not the case here. Then, for our definition of $h_t$, it is easy to see that the buyer under consideration wins the auction if and only if $s_t(V_t) \geq h_t$, as desired. In particular, as we allow $h_t$ to be chosen by an adaptive adversary (with the exception of Section~\ref{sec:log_regret}), we can capture environments where the highest competing bid is determined by competing buyers (with independent values) and the seller running their own private learning algorithms.

\section{Concave Formulation}\label{sec:convex}

Note that the space of bidding strategies is infinite dimensional and the map $s(\cdot)\mapsto u(s| F, h)$ is non-convex, which makes online optimization over the set of bidding strategies unwieldy. As our first step, we circumvent this hurdle and show that the problem of utility maximization in first-price auctions over pure strategies can be formulated as a finite-dimensional concave maximization problem in the space of bidding \emph{probabilities}. This reformulation forms the cornerstone of all our algorithms and results.

Consider a buyer who participates in a single-shot first-price auction where the set of possible bids are $0 = b_0 < b_1 < \dots < b_K \leq 1$. As before, assume that her value $v \in [0,1]$ for the item is drawn from an absolutely continuous distribution with CDF $F: [0,1] \to [0,1]$ and density $f: [0,1] \to \R_+$. Moreover, let $h$ denote the highest competing bid and assume that it's distributed according to $\pmb d = (d_0, d_1, \dots, d_K) \in \Delta^{K+1}$, independently of the value $v \sim F$. Importantly, point masses are independent of all other distributions, and thus a deterministic $h= h_t$ is independent of $v = V_t \sim F$, which is our setting of interest. Furthermore, independence of the value and the highest competing bid is a common assumption in the literature~(for example, \citealt{han2020optimal, balseiro2022contextual, wang2023learning}) on first-price auctions, and holds in practice for large-scale markets.

Let $s: [0,1] \to \{b_0, b_1, \dots, b_K\}$ be a bidding strategy of the buyer, i.e., she bids $s(v)$ when her value is $v$. 
For this strategy $s(\cdot)$ and highest competing bid distribution $\pmb d$, the expected utility $u(s | F, \pmb d)$ of the buyer is given by
\begin{align*}
	u(s | F, \pmb d) = \E_{v \sim F, h \sim \pmb d} \left[ (v - s(v)) \cdot \mathbf 1 (s(v) \geq h) \right] = \E_{v \sim F} \left[ (v - s(v)) \cdot \sum_{i: b_i \leq s(v)} d_i \right] \,.
\end{align*}

We first simplify the space of strategies by showing that it is sufficient to restrict attention to non-decreasing and left-continuous strategies that never overbid. In particular, the optimal strategy that optimizes the utility at each value always satisfies these properties (if ties are broken appropriately).

\begin{lemma}\label{lemma:non-decreasing-opt}
For each $v \in [0, 1]$, define $s^*(v)$ to be the bid $b_j \in \{b_0, b_1, \dots, b_K\}$ which maximizes the quantity $(v - b_j) \cdot \sum_{i=0}^j d_i$, choosing the smaller one in case of equality. Then $s^* \in \argmax_{s(\cdot)} u(s|F, \pmb d)$. Moreover, $s^*$ is non-decreasing, left continuous, and satisfies $s^*(v) \leq v$.

\end{lemma}

In the rest of this section, we will assume that the bidding strategy $s(\cdot)$ is non-decreasing, left-continuous and satisfies $s(v) \leq v$ for all $v \in [0,1]$. Then, if we set $v_i = \max\{v \in [0,1] \mid s(v) \leq b_{i-1}\}$ for all $0 \leq i \leq K$, we get $0 = v_0 \leq v_1\leq v_2 \leq \dots v_K \leq v_{K+1} = 1$ and $s(v) = b_i$ for all $v \in (v_i, v_{i+1}]$. In other words, we can alternately parameterize the bidding strategy $s(\cdot)$ in terms of the \emph{value thresholds} $\pmb v = \{v_i\}_{i=1}^K$ such that the bid $s(v)$ is constant between any two consecutive thresholds. In particular, note that $v_i \geq b_i$ because $b_i =s(v) \leq v$ for all $v \in (v_i, v_{i+1}]$. Now, we can rewrite the utility function $u(s | F, \pmb d)$ in terms of value thresholds as follows:
\begin{align}\label{eq:convex-refor-inter-1}
	u(s| F, \pmb d) = u(\pmb v| F, \pmb d) &\coloneqq \E_{v \sim F} \left[ \sum_{j=0}^K (v - b_j) \cdot \left(\sum_{i = 0}^j d_i \right) \cdot \mathbf{1}\left(v \in (v_j, v_{j+1}] \right) \right] \nonumber\\
	&= \sum_{j=0}^K \sum_{i = 0}^j \E_{v \sim F} \left[ (v - b_j) \cdot  d_i \cdot \mathbf{1}(v \in (v_j, v_{j+1}]) \right] \nonumber\\
	&= \sum_{i=0}^K \sum_{j=i}^K d_i \cdot  \E_{v \sim F} \left[ (v - b_j) \cdot \mathbf{1}(v \in (v_j, v_{j+1}]) \right] \nonumber \\
	&= \sum_{i=0}^K d_i \cdot \left(\E_{v \sim F}\left[ v \cdot \mathbf{1}(v > v_i) \right] - \sum_{j=i}^K b_j \cdot (F(v_{j+1}) - F(v_j)) \right) \,,
\end{align}

Note that this transformation allows use to reduce the infinite-dimensional optimization problem $\max_{s(\cdot)} u(s|F, \pmb d)$ to a finite-dimensional one $\max_{\pmb v} u(\pmb v| F, \pmb d)$, which is a considerable simplification. However, since $F$ can be arbitrary, $\pmb v \mapsto u(\pmb v| F, \pmb d)$ may not be convex, we are still left with a non-convex problem. Getting rid of this non-convexity requires yet another change of variables, which we outline next.

To motivate our approach, we first rewrite $u(\pmb v| F, \pmb d)$ in terms of the \textit{generalized inverse} of $F$, defined as $F^{-}(y) = \inf\{v \in [0,1] \mid F(v) \geq y\}$. To do so, we use the fact that, if $U$ is the uniform random variable over $[0,1]$, then $F^{-}(U)$ is distributed according to the CDF $F$. Therefore,
\begin{align*}
	\E_{v \sim F}\left[ v \cdot \mathbf{1}(v > v_i) \right] &=\E_{U}\left[ F^{-}(U) \cdot \mathbf{1}(F^{-}(U) > v_i) \right]\\
	&= \int_0^1F^-(u) \cdot \mathbf{1}(F^{-}(u) > v_i) \cdot du\\
	&= \int_0^1F^-(u) \cdot \mathbf{1}(u > F(v_i)) \cdot du\\
	&= \int_{F(v_i)}^1 F^- (u) \cdot du\,,
\end{align*}
where the third equality follows from part (5) of Proposition 1 of \citet{embrechts2013note}. This allows us to simplify \eqref{eq:convex-refor-inter-1} further and write
\begin{align}\label{eq:convex-refor-inter-2}
	u(s| F,\pmb d) = u(\pmb v| F, \pmb d) = \sum_{i=0}^K d_i \cdot \left( \int_{F(v_i)}^1 F^- (u) \cdot du - \sum_{j=i}^K b_j \cdot (F(v_{j+1}) - F(v_j)) \right) \,.
\end{align}

\begin{figure}[t!]
    \centering
    \includegraphics[width=0.5\textwidth]{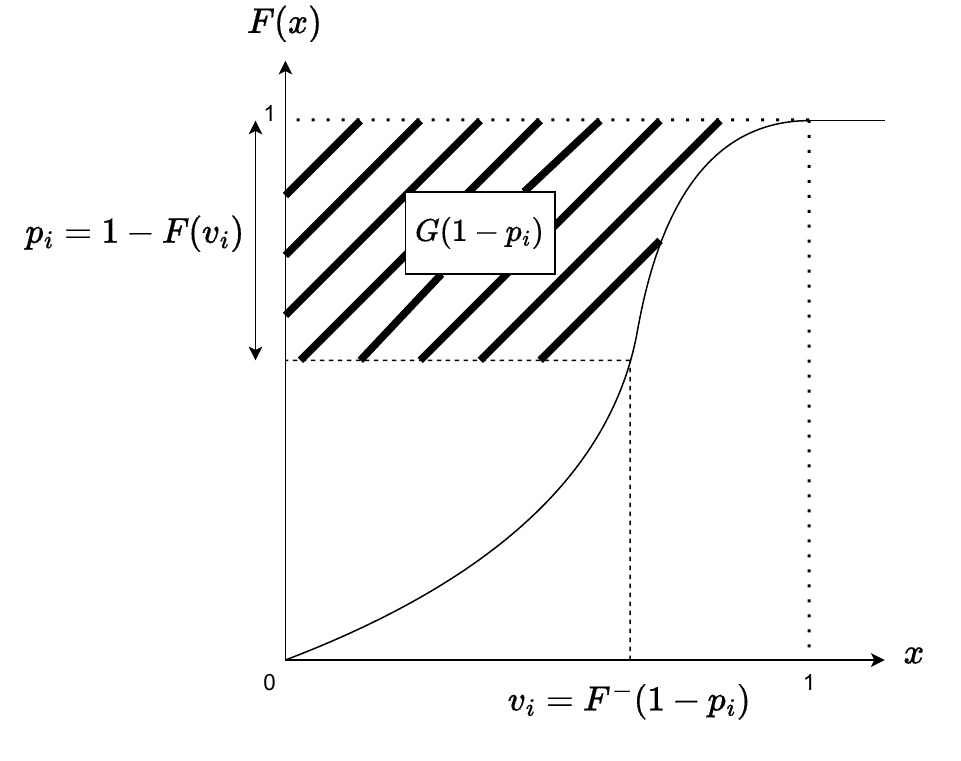}
    \caption{The change of variables that transforms a value threshold $v_i$ to the corresponding bidding probability $p_i$, and vice versa. The indicated area represents the concave integral term $G(1-p_i) = \int_{1-p_i}^1 F^-(u) \cdot du$ in the utility function $u(\p|F,\d)$.}
    \label{fig:enter-label}
\end{figure}

Now, observe that since $F^{-}(u)$ is non-decreasing, the function $G(x) \coloneqq \int_x^1 F^-(u) \cdot du$ is concave. Formally, for $0 \leq x_1 < x_2 \leq 1$ and $\bar x = (x_1 + x_2)/2$, we have
\begin{align*}
	&G(x_1) - G(\bar x) = \int_{x_1}^{\bar x} F^-(u) \cdot du \leq \int_{\bar x}^{x_2} F^-(u) \cdot du = G(\bar x) - G(x_2) \implies \frac{G(x_1) + G(x_2)}{2} \leq G(\bar x)\,,
\end{align*}
where the inequality follows from the fact that $F^-(\cdot)$ is non-decreasing. As $G(x)$ is continuous, the above inequality implies that $G(x)$ is concave.

This observation motivates our final change of variables. Let $p_i$ denote the probability that the buyer submits a bid greater than or equal to $b_i$, i.e., set $p_j = \mathbb P(s(v) \geq b_j) = 1 - F(v_j)$ for all $j \in [K]$ (define $p_{K+1} = 0$ for convenience). Then, $F(v_i) = 1 - p_i$, and we can rewrite \eqref{eq:convex-refor-inter-2} in the form
\begin{align}\label{eq:convex-refor}
	u(s| F,\pmb d) = u(\pmb p| F, \pmb d) \coloneqq \sum_{i=0}^K d_i \cdot \left( \int_{1 - p_i}^1 F^- (u) \cdot du - \sum_{j=i}^K b_j \cdot (p_j - p_{j+1}) \right)\,.
\end{align}

As $G(x) \coloneqq \int_x^1 F^-(u) \cdot du$ is concave, the function $\p \mapsto u(\p|F,\d)$ is a positive linear combination of concave functions and purely linear terms, and therefore itself is concave. Thus we have derived a concave formulation for the utility maximization problem, and (surprisingly) done so without relying on randomized bidding strategies. We summarize the reformulation result in the following theorem, which delineates the transformation between the space of bidding strategies and the space of bidding probabilities $\pp = \{\pmb p \in [0,1]^K \mid p_j \geq p_{j+1},\ p_j \leq 1 - F(b_j) \}$.

\begin{theorem}\label{thm:convex-refor}
The following statements hold for all value distributions $F$ and competing bid distributions $\pmb d \in \Delta^{K+1}$:
	\begin{enumerate}
		\item $\pmb p \mapsto u(\pmb p | F, \pmb d)$ is concave.
		\item Let $s: [0,1] \to \{b_0, b_1, \dots, b_K\}$ be a non-decreasing left-continuous bidding strategy with $s(v) \leq v$ for all $v \in [0,1]$, and set $p_j = \mathbb P(s(v) \geq b_j)$ for all $j \in [K]$. Then, $\pmb p \in \pp$ and  $u(s|F, \pmb d) = u(\p| F, \pmb d)$.
		\item Let $\pmb p \in \pp$ and define bidding strategy $s:[0,1]\to \{b_0, b_1, \dots, b_K\}$ as
			\begin{align*}
				s(v) = b_i \quad \text{for } v \in \left(F^-\left(1- p_i \right), F^-\left(1 - p_{i+1} \right) \right]\,.
			\end{align*}
			and $s(0) = 0$. Then, $u(s|F, \pmb d) = u(\p| F, \pmb d)$.
	\end{enumerate}
\end{theorem}
\begin{remark}
    Similar change-of-variables to move the problem to the ``Quantile Space" have been employed in the literature on pricing to great effect (see \citealt{hartline2013mechanism} and the references therein). However, unlike our result, the value distribution needs to be regular to achieve concavity in pricing. 
\end{remark}

Having formulated the problem of bidding in first-price auctions as a concave maximization problem, we can now exploit the powerful machinery of Online Convex Optimization (\citealt{shalev2012online, hazan2016introduction}) in order to develop learning algorithms for bidding in first-price auctions. First, in Section~\ref{sec:known_dist}, we propose and analyze a natural Online Gradient Ascent algorithm based on our concave formulation. We prove that, in addition to attaining the optimal regret scaling of $O(\sqrt{T})$, it is robust to strategic reserve pricing by the seller. However, the direct application of Online Gradient Ascent to the concave formulation requires knowledge of the value distribution $F$, which may not always be available. In Section~\ref{sec:unknown_dist}, we propose another Gradient-Ascent-based algorithm which does not require the knowledge of the value distribution $F$. It also attains $O(\sqrt{T})$-regret while being robust to strategic reserve pricing by the seller, and is additionally incentive compatible as an autobidding algorithm for the buyer.

%% file: known_distribution.tex
\section{Known Value Distribution}\label{sec:known_dist}

In this section, we will assume that the value distribution $F$ is known to the buyer ahead of time (before the first auction). Leveraging the concave formulation of Theorem~\ref{thm:convex-refor}, we propose an algorithm (Algorithm~\ref{alg:known_GA}) that runs Gradient Ascent in the space of bidding probabilities $\pp$ with reward functions $\{u(\cdot| F, h_t)\}_t$. To determine the bid in each auction, it translates the iterates $\pmb p_t \in \pp$ of Gradient Ascent to bidding strategies by using the change-of-variables equivalence established in Theorem~\ref{thm:convex-refor}.

\begin{algorithm}[H]
	\textbf{Input:} Value distribution $F$, initial iterate $\pmb p_1 \in \pp$, and step size $\eta$.\\
	\For{$t=1$ to $T$}
	{ 
		Observe value $V_t \sim F$;\\
		Bid $A_t(V_t) = b_i$ if $V_t \in \left(F^-\left(1- p_{t,i} \right), F^-\left(1 - p_{t,i+1} \right) \right]$;\\
		Observe highest competing bid $h_t$;\\
		Update $\pmb p_t$ with a Gradient Ascent step: 
		\begin{align}\label{eq:GA_update}
			\pmb p_{t+1} = \argmin_{\pmb p \in \pp} \| \pmb p - \pmb p_t^+\| \quad \text{where} \quad \pmb p_t^+ = \pmb p_t + \eta \cdot \grad u(\pmb p_t|F, h_t)
		\end{align}
    }
   \caption{Gradient Ascent with Known Value Distribution}
   \label{alg:known_GA}
\end{algorithm}

Before diving into the analysis of Algorithm~\ref{alg:known_GA}, we take a deeper look at its updates to build intuition. First observe that, for highest competing bid $h = b_i$, we can rewrite the expected utility $u(\pmb p| F, h)$ as
\begin{align}
	u(\pmb p| F, h) &=	\int_{1 - p_i}^1 F^- (u) \cdot du - \sum_{j=i}^K b_j \cdot (p_j - p_{j+1}) \nonumber \\
	&= \int_{1 - p_i}^1 F^- (u) \cdot du - b_i \cdot p_i - \sum_{j=i+1}^K (b_{j} - b_{j-1}) \cdot p_j \nonumber\\
	&= \int_{1 - p_i}^1 F^- (u) \cdot du - b_i \cdot p_i - \sum_{j=i+1}^K \epsilon \cdot p_j
\end{align}

In particular, this implies that the gradient $\nabla u(\pmb p_t | F, h_t)$ is given by
\begin{align*}
	\partial_j u(\pmb p_t | F, h_t) = \begin{cases}
		0 &\text{if } b_j < h_t\\
		F^-(1 - p_{t,j}) - b_j &\text{if } b_j = h_t\\
		-\epsilon &\text{if } b_j > h_t
	\end{cases}
\end{align*}

Ignoring the projection step, i.e., assuming $\pmb p_{t+1} = \pmb p_t^+$, we can see that Algorithm~\ref{alg:known_GA} updates $\p_t$ to
\begin{itemize}
	\item increase the probability of bidding $h_t = b_i$ because $F^-(1 - p_{t,i}) - b_i \geq 0$ for all $p_{t,i} \leq 1 - F(b_i)$,
	\item decrease the probability of bidding $b_j$ or higher for all $b_j > h_t$.
\end{itemize}

This is intuitive because $h_t$ is the optimal bid against the highest competing bid of $h_t$, and bidding strictly higher only increases the payment without increasing the chance of winning the item. Although the projection step is important and significantly complicates the analysis, we will largely ignore it here to build intuition. Importantly, as we show in Appendix~\ref{appendix:known_GA}, it is possible to execute the projection step in $O(K)$ time. In fact, we give a (quasi) closed-form expression for $\p_{t+1}$ in terms of $\p_t$ in Lemma~\ref{lemma:GA_update}.

The gradient $\nabla u(\p|F,h)$ also has an economic interpretation, wherein the $j$-th component $\partial_j u(\pmb p_t | F, h_t)$ is simply the change in utility from bidding $b_j$ instead of $b_{j-1}$:
\begin{itemize}
    \item If $b_j < h$, there is no change in utility from bidding $b_{j-1}$ instead.
    \item When $b_j = h$, bidding $b_j$ increases the utility by $F^-(1- p_j) - b_j$ when compared to the losing bid of $b_{j-1}$.
    \item If $b_j > h$, bidding $b_j$ increases the payment by $\epsilon$ in comparison to $b_{j-1}$. Since both $b_j$ and $b_{j-1}$ result in a win, this reduces the utility by $\epsilon$.
\end{itemize}

\subsection{Regret Guarantees}

We now investigate the regret guarantees of Algorithm~\ref{alg:known_GA}. Since Algorithm~\ref{alg:known_GA} is a variant of Online Gradient Ascent, it inherits the low-regret bounds enjoyed by that family of algorithms. In particular, it inherits the regret bound of Online Gradient Descent (e.g., see \citealt{shalev2012online} or \citealt{hazan2016introduction}), which we formally state in the following proposition.

\begin{proposition}\label{prop:sub-linear-regret}
	With step size $\eta$ and initial iterate $\p_1 \in \pp$, the regret of Algorithm~\ref{alg:known_GA} satisfies 
	\begin{align*}
		\regret(A|F) \leq \frac{K}{2 \eta} + 2 \eta\cdot T\,.
	\end{align*}
	In particular, setting $\eta = \sqrt{K/2T}$ yields
	\begin{align*}
		\regret(A|F) \leq 2\sqrt{2K} \cdot \sqrt{T}\,.
	\end{align*}
\end{proposition}

Proposition~\ref{prop:sub-linear-regret} shows that Algorithm~\ref{alg:known_GA} achieves $O(\sqrt{T})$ regret. The next result shows that this is the best that can be achieved by any online bidding algorithm.

\begin{proposition}\label{prop:regret-lower-bound}
	Let $A$ be any online algorithm for bidding in repeated first-price auctions, then there exists a non-adaptive deterministic sequence of highest competing bids $\{h_t\}_{t=1}^T$ such that
	\begin{align*}
		\regret(A| F) = \max_{s(\cdot)} \sum_{t=1}^T u(s| F, h_t)\ -\ \sum_{t=1}^T u(A_t| F, h_t) \geq \Omega(\sqrt{T})\,.
	\end{align*}
	This is true even when there is just one non-zero bid ($K=1$).
\end{proposition}

The proof of Proposition~\ref{prop:regret-lower-bound} uses an argument similar to the one used for establishing $\Omega(\sqrt{T})$-regret in Online Convex Optimization (e.g., see Theorem~3.5.1 of \citealt{hazan2016introduction}), which leverages the anti-concentration property of sums of i.i.d. binomial random variables. 

Proposition~\ref{prop:sub-linear-regret} highlights the power of our concave formulation (Theorem~\ref{thm:convex-refor}): it allows us to directly leverage the powerful theory of Online Convex Optimization to get the optimal regret rate. Moreover, unlike previous algorithms, our Algorithm~\ref{alg:known_GA} uses pure strategies which are monotonic. This ensures that having a higher values never leads to lower bids, a property that algorithms based on randomized strategies (like the ones proposed in \citealt{balseiro2022contextual}) lack.

\subsection{Strategic Robustness}

Although $O(\sqrt{T})$-regret is a desirable property, a variety of other algorithms proposed in previous works also achieve $O(\sqrt{T})$ regret. However, as we discussed in the introduction, regret is not the end-all-be-all performance metric, and other properties of algorithms turn out to be equally important in real-world auction markets. Specifically, the strategic nature of online advertising markets implies that resistance to manipulation by the seller is of paramount importance. In the remainder of this section, we will demonstrate the strategic robustness of Algorithm~\ref{alg:known_GA} by proving that it limits the seller's average revenue to that attained under the optimal mechanism (posting the monopoly reserve), i.e., the seller cannot exploit their knowledge of the buyer's algorithm to extract more average revenue from her than is possible under the optimal single-shot mechanism. With this, Algorithm~\ref{alg:known_GA} stands out from past work, all of which either attain $O(\sqrt{T})$ regret or are strategically robust, but fail to achieve both simultaneously.

Before proceeding with the formal statement and proof of strategic robustness of Algorithm~\ref{alg:known_GA}, we introduce some notation. When the buyer uses the bidding strategy $s(\cdot)$, the revenue that the seller extracts from her under the reserve-price/highest-competing-bid $h$ is given by
\begin{align*}
	\rev(s, h) \coloneqq \E_{v \sim F}\left[ s(v) \cdot \mathbf 1(s(v) \geq h)\right]\,.
\end{align*}
Moreover, let $\mye(F)$ denote the revenue obtained by the optimal mechanism, i.e., 
\begin{align*}
	\mye(F) = \max_{r \in [0,1]} r \cdot (1 - F(r))\,.
\end{align*}

The following theorem demonstrates the strategic robustness of Algorithm~\ref{alg:known_GA}. It states that the maximum average revenue that can be extracted from a buyer using Algorithm~\ref{alg:known_GA} is bounded above by $\mye(F) + O(1/\sqrt{T})$. In other words, if the buyer uses Algorithm~\ref{alg:known_GA} to bid and the seller wants to maximize the revenue that is extracted from her, she cannot do much better than posting the monopoly reserve price $h_t = \argmax_{r \in [0,1]} r \cdot (1 - F(r))$ at each time step $t \in [T]$. This is despite the fact that the seller has complete knowledge of the exact algorithm being used by the buyer and her value distribution. Importantly, none of the mean-based algorithms, like the popular Multiplicative/Exponential Weights algorithm, satisfy this property~\citep{braverman2018selling}; we provide a concrete example demonstrating their lack of strategic robustness in Subsection~\ref{subsec:threshold_robust}.

\begin{theorem}\label{thm:strat_robust}
	With step size $\eta = \sqrt{K/2T}$ and initial iterate $\p_1 \in \pp$, Algorithm~\ref{alg:known_GA} satisfies
	\begin{align*}
		\sum_{t=1}^T \E[\rev(A_t, h_t)] \leq \mye(F) \cdot T + \sqrt{2KT}\,.
	\end{align*}
\end{theorem}

\begin{remark}
    In fact, our proof yields a stronger instance-dependent upper bound. For $h_t = b_{i_t}$ for all $t \in [T]$, we get:
    \begin{align*}
        \sum_{t=1}^T \E[\rev(A_t, h_t)] \leq \sum_{t=1}^T p_{t, i_t} \cdot F^-(1 - p_{t, i_t}) + \sqrt{2KT}\,.
    \end{align*}
    Note that $p_{t, i_t} \cdot F^-(1 - p_{t, i_t})$ is simply the revenue attained from the posted-price mechanism that sells the item with probability exactly $p_{t, i_t}$, which is the probability that the algorithm wins the item in auction $t$. Therefore, $\sum_{t=1}^T p_{t, i_t} \cdot F^-(1 - p_{t, i_t})$ is the total revenue attained from selling the $T$ items separately to a strategic buyer, with the price for item $t$ being $F^-(1 - p_{t, i_t})$. Under this mechanism, the probability of selling item $t$ is equal to $p_{t, i_t}$, which is the probability with which item $t$ is sold to Algorithm~\ref{alg:known_GA}, thereby yielding a more fine-grained instance-dependent bound.
\end{remark}

The proof of Theorem~\ref{thm:strat_robust} is based on a novel potential-function argument that couples the revenue of the seller in each auction with the change in the potential caused by the update of Algorithm~\ref{alg:known_GA}. Importantly, unlike \citet{braverman2018selling}, it does not rely on ex-post incentive compatibility, and instead directly establishes strategic robustness, which may be of independent interest. The intuition of the argument is best conveyed using the continuous-time approximation of Gradient Ascent, and that is what we present here. The full proof is more involved because it has to contend with discrete updates and the projection onto $\pp$; it can be found in Appendix~\ref{appendix:known_GA}.

Let $\p \in \pp$ be the bidding probabilities of the buyer. Define the potential $\Phi(\p)$ to be the squared Euclidean norm of $\p$ scaled down by $1/2$, i.e.,
\begin{align*}
    \Phi(\p) \coloneqq \frac{\|\p\|^2}{2} =  \frac{\sum_{j=1}^K p_j^2}{2}
\end{align*}
When $h = b_i$, note that the revenue obtained from the bidding strategy corresponding to bidding probabilities $\p$ (Theorem~\ref{thm:convex-refor}) is given by
\small
\begin{align*}
    \rev(\p, h) = \sum_{j=i}^K b_j \cdot \mathbb{P}(\text{Bid} = b_j) =  \sum_{j=i}^K b_j \cdot (p_j - p_{j+1}) = b_i \cdot p_i + \sum_{j=i+1}^K (b_j - b_{j-1}) \cdot p_j = b_i \cdot p_i + \epsilon \cdot \sum_{j=i+1}^K p_j\,.
\end{align*} 
\normalsize
Moreover, the continuous-time approximation of the update step of Algorithm~\ref{alg:known_GA} is given by
\begin{align*}
    \frac{d\p}{dt} = \nabla u(\p|F, h) \quad \text{where} \quad \partial_j u(\pmb p | F, h) = \begin{cases}
		0 &\text{if } j < i\\
		F^-(1 - p_{i}) - b_i &\text{if } j = i\\
		-\epsilon &\text{if } j > i
	\end{cases}
\end{align*}
We start by showing that the excess revenue over $\mye(F)$ accrued by the seller can be charged against the change in potential $d\Phi(\p)/dt$ and $\mye(F)$. To see this, observe that the change in potential is given by
\begin{align*}
    \frac{d\Phi(\p)}{dt} = \sum_{j=1}^K \frac{\partial \Phi(\p)}{\partial p_j} \cdot \frac{dp_j}{dt} =  \sum_{j=1}^K p_j \cdot \frac{dp_j}{dt} = p_i \cdot (F^-(1 - p_{i}) - b_i) - \epsilon \cdot \sum_{j=i+1}^K p_j
\end{align*}

Therefore, we have
\begin{align*}
     \frac{d\Phi(\p)}{dt} + \rev(\p, h) &= p_i \cdot (F^-(1 - p_{i}) - b_i) - \epsilon \cdot \sum_{j=i+1}^K p_j + b_i \cdot p_i + \epsilon \cdot \sum_{j=i+1}^K p_j\\
    &= p_i \cdot F^-(1- p_i)\\
    &\leq \mye(F)\,,
\end{align*}
where the last inequality follows from $\mye(F) \geq r \cdot (1 - F(r))$ for $r = F^-(1 - p_i)$. As the reserve price $h = b_i$ was arbitrary, we have argued that the excess revenue over $\mye(F)$ can be charged against the change in potential, i.e.,
\begin{align*}
    \rev(\p, h) - \mye(F) + \frac{d\Phi(\p)}{dt} \leq 0\,.
\end{align*}
Integrating over time $t$ and using $\rev$ to denote the total expected revenue of the seller yields
\begin{align*}
    \rev - \mye(F) \cdot T + \Phi(\p_T) - \Phi(\p_1) \leq 0\,.
\end{align*}

As $\Phi(\p) - \Phi(\tilde \p) \leq K$ for any $\p, \tilde \p \in \pp$, we get that the total revenue extracted from the buyer is at most a constant larger than $\mye(F) \cdot T$. The analysis of the continuous-time approximation is not exact, and going from continuous time to discrete time introduces the additional $O(\sqrt{T})$ error that appears in the guarantee of Theorem~\ref{thm:strat_robust}. The simplicity of the aforementioned potential function argument is worth emphasizing; one rarely comes across a potential function simpler than the Euclidean norm. It yet again highlights the utility of our concave formulation in drastically simplifying the design and analysis of algorithms bidding first-price auctions. Contrast it with how one might go about implementing and analyzing Gradient Ascent in the absence of a concave formulation. Natural attempts would likely involve discretizing the value space and running an independent copy of Gradient Ascent for each value on the space of probability distributions over bids (like \citealt{balseiro2022contextual} do with Exponential Weights). This approach quickly runs into serious challenges. Due to the discretization, it either results in high (scaling with $T$) computational/memory utilization at each time period or suffers from sub-optimal regret. More importantly, it is difficult to imagine a simple argument establishing strategic robustness for such a complicated algorithm with multiple independent copies of Gradient Ascent. Thus, we consider our simple implementation of Gradient Ascent, which does not require multiple independent copies for each value, to be an important contribution in its own right---one that applies to all Online Convex Optimization algorithms.

In this section, we assumed that the value distribution $F$ was known to the algorithm designer. Since this may not always be the case in practice, we relax this assumption in the next section, and develop an algorithm that does not require knowledge of $F$.

%% file: unknown_dist.tex
\section{Unknown Value Distribution}\label{sec:unknown_dist}

\begin{algorithm}
	\textbf{Input:} Initial iterate $\v_1 \in \vv$ and step size $\eta$.\\
	\For{$t=1$ to $T$}
	{ 
		Observe value $V_t \sim F$;\\
		Bid $A_t(V_t) = b_i$ if $V_t \in (v_{t,i}, v_{t,i+1}]$ (where $v_{K+1} = 1$);\\
		Observe highest competing bid $h_t$;\\
		Update $\v_{t+1} = \argmin_{\v \in \vv} \| \v- \v_t^+\|$, where
		\begin{align}\label{eq:threshold_update}
			 v_{t,i}^+ = \begin{cases}
			 	v_{t,i} + \eta \cdot \epsilon &\text{if } b_i > h_t\\
			 	v_{t,i} - \eta \cdot (v_{t,i} - h_t) &\text{if } b_i = h_t\\
			 	v_{t,i} &\text{if } b_i < h_t
			 \end{cases}
		\end{align}
    }
   \caption{Value-Threshold-Based Algorithm}
   \label{alg:threshold}
\end{algorithm}

In this section, we no longer assume that the value distribution $F$ is known to the algorithm-designer, and instead only assume knowledge of an upper bound $\bar f \geq \sup_{x \in [0,1]} f(x)$ on the density. The knowledge of $\bar f$ is not crucial to our results: we use it solely for tuning the step size $\eta$, and setting $\eta = 1/\sqrt{T}$ yields the same dependence on $T$ even in its absence. Now, as the utility function $\p \mapsto u(\p|F,h)$ and its gradient $\p \mapsto \nabla u(\p|F, h)$ depend on the value distribution $F$, we can no longer directly implement Gradient Ascent to our concave formulation. However, it turns out that the design principles of Algorithm~\ref{alg:known_GA} continue to work well even in this setting, and we use them to construct Algorithm~\ref{alg:threshold}. In particular, Algorithm~\ref{alg:threshold} maintains feasible value-thresholds $\v_t \in \vv$, where 
\begin{align*}
	\vv \coloneqq \{\v \in [0,1]^K \mid v_i \leq v_{i+1},\ v_i \geq b_i  \}\,,
\end{align*} 
\noindent
and bids $b_i$ in auction $t$ whenever $v \in (v_{t,i}, v_{t,i+1}]$. It updates the value-thresholds iteratively with the same directional changes as Algorithm~\ref{alg:known_GA}: upon encountering the highest competing bid $h_t = b_k$ for some $k \in [K]$, it decreases the probability of bidding strictly greater than $b_k$ by increasing the thresholds $v_i$ for all $i > k$, and it increases the probability of bidding $b_k$ by decreasing the threshold $v_k$; leaving all other thresholds unchanged. Even though Algorithm~\ref{alg:known_GA} and Algorithm~\ref{alg:threshold} may seem very different at first sight, the following proposition demonstrates their intimate connection: Algorithm~\ref{alg:threshold} is simply Algorithm~\ref{alg:known_GA} run with the uniform value distribution.

\begin{proposition}\label{prop:alg_equivalence}
	Fix the sequence of highest competing bids $\{h_t\}_{t=1}^T$. Let $\{\p_t\}_{t=1}^T$ be the sequence of iterates produced by Algorithm~\ref{alg:known_GA} with the uniform value distribution as the input (i.e., $F$ with $F(x) = x$ for all $x \in [0,1]$), and let $\{\v_t\}_{t=1}^T$ be the iterates produced by Algorithm~\ref{alg:threshold}. If the the initial iterates $\v_1, \p_1$ satisfy $\v_1 = \pmb 1 - \p_1$, then we have $\v_t = \pmb 1 - \p_t$ for all $t \in [T]$.
\end{proposition}

\subsection{Regret Guarantee}

Algorithm~\ref{alg:threshold} is motivated by Proposition~\ref{prop:alg_equivalence}---it is not possible to run Algorithm~\ref{alg:known_GA} without knowledge of the value distribution $F$, so we \emph{pretend} that the value distribution is the uniform distribution over $[0,1]$ and deploy Algorithm~\ref{alg:known_GA} with the uniform distribution. However, this mismatch between the true value distribution $F$ and the uniform distribution breaks the concavity that has hitherto driven our analysis---the utility $u(\v|F,h)$ of the buyer as a function of the value thresholds $\v$ is not always concave. In particular, this means that we cannot use Proposition~\ref{prop:sub-linear-regret}, or extend its proof to establish regret bounds for Algorithm~\ref{alg:threshold}. Despite the lack of concavity, Algorithm~\ref{alg:threshold} performs well and attains $O(\sqrt{T})$-regret.

\begin{theorem}\label{thm:regret_threshold}
	With step-size $\eta$ and initial iterate $\v_1 \in \vv$, Algorithm~\ref{alg:threshold} satisfies
	\begin{align*}
		\regret(A|F) \leq 6\bar f K \cdot \eta T + \frac{K}{\eta}\,\,.
	\end{align*}
	In particular, with $\eta = 1/\sqrt{\bar f T}$, we get
	\begin{align*}
		\regret(A|F) \leq 7 {\bar f}^{\frac{1}{2}} K \cdot \sqrt{T}
	\end{align*} 
\end{theorem}

Since the utility may not be concave in the value thresholds, we can no longer leverage results from Online Convex Optimization like we did in the analysis of Algorithm~\ref{alg:known_GA}; a new approach is required. The proof of Theorem~\ref{thm:regret_threshold} provides such a new approach in the form of a novel potential function argument which does not use any of the machinery from Online Convex Optimization, and can be found in Appendix~\ref{appendix:unknown}. The central insight lies in the careful design of the potential function, which allows us to charge the change in regret to the change in potential. Here, we present a proof-sketch for the special case where $K=1$, i.e.,  there is only 1 possible non-zero bid, namely $b_1 = \epsilon < 1$. Once again, the intuition behind the proof is best conveyed with a continuous-time argument that allows us to ignore the difficult (and tedious) edge cases.

Fix the benchmark policy $s^*:[0,1] \to \{b_0, b_1\}$ against which we want low regret (assume no overbidding). Moreover, also fix a value $v^* \in [0,1]$; we will prove the no-regret property for each value separately. Let $v_1$ denote the threshold such that the buyer bids $b_1$ when the value $V > v_1$ and bids $0$ otherwise. Define the following potential function:
\begin{align*}
    \Phi(v_1|v^*) = v_1 \cdot \left[\mathbf{1}(s^*(v^*) = b_1) - \mathbf{1}(v^* > v_1) \right]
\end{align*}
We will show that the regret corresponding to value $v^*$ can be charged against the change in the potential $\Phi(v_1|v^*)$ for any possible highest competing bid $h$.  First, observe that the continuous-time update of Algorithm~\ref{alg:threshold} is given by
\begin{align*}
    \frac{dv_1}{dt} = \begin{cases}
        -(v_1 - b_1) &\text{if } h = b_1\\
        \epsilon &\text{if } h = b_0\,.
    \end{cases}
\end{align*}

Note that the algorithm only accumulates regret when it bids something different from the benchmark $s^*$. Moreover, it also has zero regret when $v^* \leq b_1$: the algorithm always bids $b_0$ for value $v^*$ because $v_1 \geq b_1$, which is the same as the benchmark $s^*$. Therefore, assume $v^* > b_1$. Thus, for the highest competing bid $h$, it only accumulates non-zero regret in the following cases,
\begin{itemize}
    \item $s^*(v^*) = h = b_0$ and $v_1 < v^*$: In this case, the algorithm bids $b_1$ for value $v^*$, whereas the benchmark bids $s^*(v^*) = 0$. As $h=b_0$, both bids result in a win, but the benchmark pays $b_1 - b_0 = \epsilon$ less than the algorithm, thereby incurring $\epsilon$ regret. On the other hand, in this case we have $d\Phi(v_1|v^*)/dt = -\epsilon$.
    \item $s^*(v^*) = h = b_1$ and $v_1 \geq v^*$: In this case, the algorithm bids $b_0$ for value $v^*$, whereas the benchmark bids $s^*(v^*) = b_1$. As $h=b_1$, the algorithm does not win, but the benchmark does, resulting in a regret of $v^* - b_1$. On the other hand, in this case we have $d\Phi(v_1|v^*)/dt = - (v_1 - b_1) \leq -(v^* - b_1)$.
    \item $s^*(v^*) = b_0 < h = b_1$ and $v_1 < v^*$:  In this case, the algorithm bids $b_1$ for value $v^*$, whereas the benchmark bids $s^*(v^*) = 0$. As $h=b_1$, the algorithm wins and the benchmark loses, resulting in a negative regret of $-(v^* - b_1)$. On the other hand, in this case we have $d\Phi(v_1|v^*)/dt = v_1 - b_1 \leq v^* - b_1$.
    \item $s^*(v^*) = b_1 > h = b_0$ and $v_1 \geq v^*$: In this case, the algorithm bids $b_0$ for value $v^*$, whereas the benchmark bids $s^*(v^*) = b_1$. As $h=b_0$, both bids result in a win, but the benchmark overpays by $\epsilon$, resulting in a negative regret of $-\epsilon$. On the other hand, in this case we have $d\Phi(v_1|v^*)/dt = \epsilon$.
\end{itemize}
Therefore, in all cases, we have shown that the regret for value $v^*$ can be charged against the change in potential, i.e.,
\begin{align*}
    \frac{d\Phi(v_1|v^*)}{dt} + \regret(v_1|v^*, h) \leq 0\,,
\end{align*}
where $\regret(v_1|v^*, h)$ is the regret for value $v^*$ associated with bidding according to the threshold $v_1$ instead of the benchmark $s^*(v^*)$, when the highest competing bid is $h$. Integrating over time $t$ and $v^*$, and using $\regret$ to denote the total expected regret yields
\begin{align*}
    \Phi(v_T|v^*) - \Phi(v_1|v^*) + \regret \leq 0\,.
\end{align*}
As $-1 \leq \Phi(x) \leq 1$ for all $x$, we get that $\regret$ is at most 2. This analysis of the continuous-time approximation is inexact and simplified. The proof becomes much more intricate when one goes to discrete time and considers the general case of $K > 1$, which introduces the additional $O(\sqrt{T})$ error that appears in the guarantee of Theorem~\ref{thm:regret_threshold}.

\subsection{Strategic Robustness}\label{subsec:threshold_robust}

Having established a $O(\sqrt{T})$-regret guarantee for Algorithm~\ref{alg:threshold}, we now turn our attention to its strategic robustness. In the known-distribution setting, the allure of Algorithm~\ref{alg:known_GA} over previously-proposed algorithms stems from its added ability to limit the average revenue extracted from the buyer to $\mye(F)$. The next theorem shows that Algorithm~\ref{alg:threshold} is also robust to strategic manipulation by the seller, without requiring knowledge of the value distribution $F$ (aside from a bound on $\bar f$). 

\begin{theorem}\label{thm:threshold_robust}
	With step size $\eta = 1/\sqrt{\bar f T}$ for some constant $a> 0$ and initial iterate $\v_1 \in \vv$, Algorithm~\ref{alg:threshold} satisfies
	\begin{align*}
		\sum_{t=1}^T \E[\rev(A_t, h_t)] \leq \mye(F) \cdot T + 2 \bar f^{\frac{1}{2}} K \cdot \sqrt{T} \,\,.
	\end{align*}
\end{theorem}

The proof of Theorem~\ref{thm:threshold_robust} is based on a carefully-designed novel potential function that allows us to charge the change in revenue against the change in potential. We omit the proof here and instead discuss a concrete example in which mean-based algorithms (like Multiplicative Weights and FTRL) yield a revenue higher than $\mye(F)\cdot T$, but our Algorithm~\ref{alg:threshold} does not.

\begin{example*}
    Consider a single buyer whose value distribution is a smoothly-truncated equi-revenue distribution starting at 1/8, i.e., 
    \begin{align*}
        F(x) = \begin{cases}
            0 &\text{if } x \leq 1/8\\
            1 - \frac{1}{8x} &\text{if } 1/8 < x < 1 - \delta\\
            1 - \frac{1-x}{8(1 - \delta)\delta} &\text{if } x \geq 1 - \delta
        \end{cases}
    \end{align*}
    for some small constant $\delta \in (0,0.5)$. The possible bids are $b_0 = 0$, $b_1 = 1/8$ and $b_2 = 1/4$. It is straightforward to check that posting a price of either $b_1 = 1/8$ or $b_2 = 1/4$ leads to a revenue of $1/8$. In fact, 1/8 is the maximum revenue that can be achieved by any price because
    \begin{align*}
         x \cdot (1 - F(x)) = \begin{cases}
            x &\text{if } x \leq 1/8\\
            \frac{1}{8} &\text{if } 1/8 < x < 1 - \delta\\
            \frac{x(1-x)}{8(1 - \delta)\delta} &\text{if } x \geq 1 - \delta
        \end{cases}
    \end{align*}
    Therefore, we have $\mye(F) = 1/8$.
    
    Consider the sequence of decreasing reserve prices $\{h_t\}_t$ such that $h_t = b_2 = 1/4$ for $t \leq T/2$ and $h_t = b_1 = 1/8$ for $t> T/2$, i.e., the seller posts a reserve price of $1/4$ for the first half of the auctions and then reduces it to $1/8$ for the second half. We start by showing that this simple sequence of reserve prices is sufficient to exploit mean-based algorithms and extract more revenue than $\mye(T)\cdot T$. Informally speaking, an algorithm is mean-based if it plays historically sub-optimal actions with a small probability (see \citealt{braverman2018selling} for a formal definition). In other words, they almost always play actions that yield the highest historical cumulative utility. Many popular algorithms like Exponential Weights, EXP3 and FTRL are mean based, and consequently most of the recently proposed algorithms for bidding in first-price auctions are also mean based (see Subsection~\ref{subsec:related} for a discussion). 

    First note that, in the first $T/2$ auctions, bidding $1/4$ for values $v \geq 1/4$ is the optimal strategy for the past, i.e., maximizes the historical cumulative utility in auctions 1 through $t-1$. Therefore, in the first $T/2$ auctions, every mean-based algorithm bids $1/4$ for values $v \geq 1/4$ and bids arbitrarily for the other values. Importantly, even after the shift to reserve price $1/8$ (auctions $t > T/2$), the bid with the highest historical cumulative utility remains $b_2 = 1/4$ for values $v \geq 1/2$. For values $1/4 \leq v \leq 1/2$, the bid with the highest historical cumulative utility transitions from $1/4$ to $1/8$ at some time $t \in [T/2 + 1, T]$. Lastly, for values $v \leq 1/4$, the bid with the highest historical utility is $1/8$. Therefore, in the last $T/2$ auctions, every mean-based algorithm continues to bid $1/4$ for values $v \geq 1/2$, transitions to bidding $1/4$ for values $1/4 \leq v \leq 1/2$, and bids $1/8$ for values $v \leq 1/4$. Crucially, this implies that the total payment made by any mean-based algorithm is at least
    \begin{align*}
        \underbrace{\frac{1}{4}\cdot (1 - F(1/4)) \cdot \frac{T}{2}}_{t \leq T/2}\ +\ \underbrace{\frac{1}{4}\cdot (1 - F(1/2)) \cdot \frac{T}{2}}_{t > T/2 \text{ and } v \geq 1/2}\ +\ \underbrace{\frac{1}{8} \cdot F(1/2) \cdot \frac{T}{2}}_{t > T/2 \text{ and } v < 1/2}\ \geq\ \mye(F) \cdot T + \frac{T}{64}\,.
    \end{align*}

    On the other hand, except for an initial transition period of length  $O(\sqrt{T})$, Algorithm~\ref{alg:threshold} bids $1/4$ for all values $v \geq 1/4$ in the first $T/2$ auctions. Moreover, except for a transition period of length $O(\sqrt{T})$ after the change of reserve price from $1/4$ to $1/8$, Algorithm~\ref{alg:threshold} bids $1/8$ for all values $v \geq 1/8$ in the last $T/2$ auctions. Therefore, the total payment of Algorithm~\ref{alg:threshold} is bounded above by
    \begin{align*}
        \underbrace{\frac{1}{4}\cdot (1 - F(1/4)) \cdot \frac{T}{2}}_{t \leq T/2}\ +\ \underbrace{\frac{1}{8}\cdot (1 - F(1/2)) \cdot \frac{T}{2}}_{t > T/2 \text{ and } v \geq 1/2}\ +\ \underbrace{\frac{1}{8} \cdot F(1/2) \cdot \frac{T}{2}}_{t > T/2 \text{ and } v < 1/2}\ =\ \mye(F) \cdot T\,.
    \end{align*}
    The above decomposition and comparison of total payment precisely highlights a weakness of mean-based algorithms: they are not agile and put too much weight on the distant past. In particular, they fail to learn the new optimal bid for values $v \geq 1/2$ sufficiently fast after the change in reserve price from $1/4$ to $1/8$, and this results in unnecessarily high payments for those values. In contrast, Algorithm~\ref{alg:threshold} is based on Gradient Ascent and quickly switches to the optimal bid of $1/8$ after the transition. The lack of agility on the part of mean-based algorithm not only results in higher revenue for the seller, but also lower utility for the buyer. We will use this fact to demonstrate the lack of incentive compatibility in mean-based algorithms when we continue this example in the next subsection.
\end{example*}

\subsection{Incentive Compatibility}

In the previous subsection, we showed that Algorithm~\ref{alg:threshold} is resistant to manipulation by the seller. However, thus far we have paid very little attention to manipulation by the buyer. In particular, bidding algorithms of the type developed in this paper are deployed as automated bidding algorithms (or \emph{autobidders} for short) on internet platforms. These autobidders take as input the high-level objectives of the advertiser and attempt to maximize total utility according to those objectives. One of the main inputs provided by each advertiser is her value-per-click and targeting criteria, which is used to compute her value for winning each auction. Therefore, even though buyers cannot directly choose their bids in each auctions, they can misreport their values in an attempt to gain higher utility. 

In particular, a strategic buyer can misreport their high-level objectives to the autobidder in a way that causes it to believe that her value is $M(v)$ whenever her true value is $v$. This misreporting of values is detrimental to both the buyer and the seller. The buyer has to spend effort and incur costs in order to find beneficial misreports. This in turn makes the system unpredictable for the seller and she loses the ability to measure the true value of the buyer, which is very valuable for experimentation. Thus, it is practically desirable to employ algorithms which are resistant to manipulation by a strategic buyer who has the power to misreport her values. Formally, we want algorithms that are incentive compatible: the buyer should not regret truthfully reporting her values. Like strategic robustness, mean-based algorithms fail to hit the mark here too and are not incentive compatible. In contrast, as the next theorem establishes, Algorithm~\ref{alg:threshold} is incentive compatible.

\begin{theorem}\label{thm:ic_threshold}
	For any misreport map $M:[0,1] \to [0,1]$ and initial iterate $\v_1 \in \vv$, Algorithm~\ref{alg:threshold} with step-size $\eta = 1/\sqrt{\bar f T}$ satisfies
	\begin{align*}
		\sum_{t=1}^T \E[u(s_{\v} \circ M|F,h_t)] - \sum_{t=1}^T \E[u(s_{\v}|F,h_t)] \leq 8K\bar f^{\frac{1}{2}} \cdot \sqrt{T}\,\,.
	\end{align*}
\end{theorem}

The proof of Theorem~\ref{thm:ic_threshold} is also based on a potential function argument, with the main technical component being the design of an intricate potential function and the subsequent charging argument based on that potential. We refer the reader to Appendix~\ref{appendix:unknown} for the full proof, and instead continue with our discussion of the concrete example we introduced in the previous subsection. In particular, we show that the same example demonstrates the lack of incentive compatibility in mean-based algorithms.

\begin{example*}[Continued from Subsection~\ref{subsec:threshold_robust}]
    Recall that any mean-based algorithm bids $1/4$ for all values $v \geq 1/2$ in all auctions $t \in [T]$. On the other hand, for the values $v$ close to $1/4$ (i.e., $1/4 \leq v \leq 1/4 + o(1)$), it bids nearly optimally: $1/4$ in the first $T/2$ auctions and $1/8$ in the last $T/2$ auctions, except for a short $o(1)$ transition period when the reserve price changes from $1/4$ to $1/8$. As a consequence, the buyer would receive higher expected utility by misreporting her value to be (close to) $1/4$ whenever her true value is larger than $1/2$, i.e., mean-based algorithms incentivize the buyer to misreport her value in this example.

    In contrast, recall that Algorithm~\ref{alg:threshold} bids $1/4$ for values $v \geq 1/4$ in the first half of the auctions, and bids $1/8$ for values $v \geq 1/8$ in the last half (ignoring the transition periods of length $O(\sqrt{T})$). Therefore, every value bids nearly optimally in all auctions, and consequently the buyer does not gain anything from misreporting her values. 
    
    Intuitively, the lack of incentive compatibility of mean-based algorithms stems from their inability to learn effectively across values: even though they learn to bid optimally for values $v$ close to $1/4$, they are not able to leverage it for larger values $1/2\leq v \leq 1$. Algorithm~\ref{alg:threshold} does not suffer from this issue. It uses the threshold structure of the bidding strategies to learn the optimal bid for all values after the change in reserve price from $1/4$ to $1/8$. In particular, the threshold $v_{t,2}$ increases to 1 within $O(\sqrt{T})$ auctions of the change in reserve price, and Algorithm~\ref{alg:threshold} only bids $1/4$ for values $v > v_{t,2}$, which results in optimal bids for all values.

    It is worth noting that this lack of incentive compatibility arises naturally. A strategic seller who wants to maximize revenue from a mean-based algorithm is incentivized to post decreasing reserve prices. And it is precisely for that sequence of reserve prices that the buyer can gain for misreporting her value to the mean-based algorithm. In other words, the sequences of reserve prices that help the seller maximize her revenue are exactly the ones which render the mean-based algorithm non-incentive-compatible.
\end{example*}

\subsection{Multi-Buyer Strategic Robustness}

In this subsection, we show that the strategic robustness of Algorithm~\ref{alg:threshold} continues to hold in the multi-agent setting where all of the buyers simultaneously employ it to bid. Consider a setting with $n$ buyers who participate in $T$ sequential second price auctions. We will use $F_i$ to denote the value distribution of buyer $i \in [n]$ (and assume that $\bar f$ is an upper bound on the density of all of the $F_i$). Let $\mye(\{F_i\}_i)$ denote the maximum revenue that can be extracted from these buyers in a single-item incentive-compatible mechanism. We assume that ties are broken based on some random ranking of the buyers; see Section~\ref{sec:tie-breaking} for the formal definition of the multi-buyer setup. When all of the buyers use Algorithm~\ref{alg:threshold} to bid, the following theorem proves that the maximum average revenue that the seller can extract from them is at the most $\mye(\{F_i\}_i)$. In other words, the seller cannot exploit Algorithm~\ref{alg:threshold} even when all of the buyers simultaneously use it, thereby extending Theorem~\ref{thm:threshold_robust} to the multi-buyer setting.

\begin{theorem}\label{thm:multi_buyer_robust}
	If all $n$ buyers employ Algorithm~\ref{alg:threshold} with $\eta = 1/\sqrt{\bar fT}$, then the total expected revenue $\rev(A, \{F_i\})$ satisfies
	\begin{align*}
		\rev(A, \{F_i\}_i) \leq \mye(\{F_i\}_i) + 8nK\bar f^{\frac{1}{2}} \cdot \sqrt{T}\,.
	\end{align*}
\end{theorem}
\begin{remark}
    An analogue of Theorem~\ref{thm:multi_buyer_robust} continues to hold even if the buyers use different step sizes. In particular, as long as each buyer uses a step size $\eta$ satisfying $\eta = \Theta(1/\sqrt{T})$, our analysis guarantees $\rev(A, \{F_i\}_i) \leq \mye(\{F_i\}_i) + O(\sqrt{T})$.
\end{remark}

The proof of Theorem \ref{thm:multi_buyer_robust} follows from the incentive compatibility guarantees we proved as part of Theorem \ref{thm:ic_threshold}. In particular, we show that the expected allocation and payment rules resulting from all bidders running Algorithm \ref{alg:threshold} form a mechanism that is close to an ex-ante truthful mechanism, which allows us to upper bound its expected revenue by the revenue of the optimal mechanism. 

We conclude by noting that the regret guarantee (Theorem~\ref{thm:regret_threshold}) and the incentive-compatibility property (Theorem~\ref{thm:ic_threshold}) hold for adaptively adversarial highest competing bids, which includes as a special case the highest competing bids generated by the simultaneous use of Algorithm~\ref{alg:threshold} by all buyers. Altogether, even though our exposition has focused on a single-buyer's perspective, the multi-buyer setting, where all of the buyer simultaneously use Algorithm~\ref{alg:threshold}, falls well-within the purview of our results, which apply even in much less-structured adversarial environments.


%% file: log_regret.tex
\section{Logarithmic Regret for Stochastic Environments}\label{sec:log_regret}

Thus far we have focused our attention on the worst-case setting where the highest competing bids are generated adversarially. In this section, we consider a more well-behaved environment and assume that the highest competing bid $h_t$ is drawn from an unknown distribution $\d$. Moreover, we relax our assumption that $b_i = \epsilon \cdot i$ for all $i \in [K]$ and allow the set of possible bids $0 = b_0 < b_1 < b_2< \dots < b_K \leq 1$ to be arbitrary. This allows us to discard highest-competing bids with zero probability of occurring and posit the existence of a positive lower bound $0 < d_\min < \min_i d_i$. Furthermore, we assume that such a lower bound $d_{\min}$ is known to the algorithm designer. In this setting, our concave formulation (Theorem~\ref{thm:convex-refor}) yields a strongly-concave reward function, thereby allowing us to attain a $O(\log T)$ regret guarantee.

\begin{proposition}\label{prop:strong_concavity}
	The utility function $\p \mapsto u(\p|F,\d)$ is $\alpha$-strongly concave for $\alpha = d_\min/\bar f$.
\end{proposition}

It is known that Stochastic Gradient Descent achieves $O(\log T)$ regret for minimizing strongly-convex functions~\citep{hazan2006logarithmic}. Thus, we can leverage Proposition~\ref{prop:strong_concavity} to obtain the following $O(\log T)$ regret guarantee for Algorithm~\ref{alg:known_GA}.

\begin{theorem}\label{thm:log_regret}
	Algorithm~\ref{alg:known_GA} with variable step size $\eta_t = \bar f/ (d_\min t)$ (see Appendix~\ref{appendix:log_regret} for a formal definition) satisfies
	\begin{align*}
		\max_{s^*(\cdot)}\ \sum_{t=1}^T u(s^*|F,\d) - \sum_{t=1}^T u(A_t|F,\d) \leq \frac{2 \bar f}{d_\min} \cdot (1 + \log T)\,.
	\end{align*}
\end{theorem}
\begin{remark}
    \citet{han2020optimal} prove a $\Omega(T)$ lower bound on regret for the setting where the value $V_t$ is constant and identical for all auctions $t \in [T]$. Theorem~\ref{thm:log_regret} shows that their lower bound does not hold for the ``smoothed" variant of the problem where the value distribution has a bounded density.
\end{remark}

This simple and easy-to-prove regret bound is a testament to the power of our concave reformulation. Our formulation highlights the hidden concavity of the problem and allows for the use of techniques from online/stochastic convex optimization, immediately giving us a $O(\log T)$-regret bound, which is an exponential improvement over the previous best of $O(\sqrt{T})$.\footnote{Both \citet{balseiro2022contextual} and \citet{han2020optimal} both give algorithms that achieve $O(\sqrt{T})$-regret under weaker feedback and arbitrary $F$. Theorem~\ref{thm:log_regret} shows that a drastic improvement can be achieved in our setting where $F$ is absolutely continuous and $h_t$ is revealed after each auction, both of which are common in practice.}

%% file: appendix_concave.tex
\section{Proofs for Section~\ref{sec:convex}}

\subsection{Proof of Lemma~\ref{lemma:non-decreasing-opt}}

\begin{proof}
	$s^* \in \argmax_{s(\cdot)} u(s|F, \pmb d)$ and $s^*(v) \leq v$ follow directly from the definition of $s^*$. Consider two values $0 \leq v < v' \leq 1$, and let $s^*(v) = b_j$, $s^*(v') = b_{j'}$. For contradiction, assume $b_j > b_{j'}$ (or equivalently, $j > j'$). Then, the definition of $s^*(\cdot)$ implies that
	\begin{align*}
		(v - b_j) \cdot \sum_{i =0}^j d_i > (v - b_{j'}) \cdot \sum_{i =0}^{j'} d_i \quad \text{and} \quad (v' - b_{j'}) \cdot \sum_{i =0}^{j'} d_i \geq (v' - b_{j}) \cdot \sum_{i =0}^{j} d_i\,
	\end{align*}
	respectively. The former is a strict inequality because ties are broken in favor of smaller bids. Adding the two inequalities together and cancelling the terms $b_j \cdot \sum_{i \leq j} d_i$ and $b_{j'} \cdot \sum_{i \leq j'} d_i$ yields
	\begin{align*}
		v \cdot \sum_{i =0}^j d_i + v' \cdot \sum_{i =0}^{j'} d_i > v \cdot \sum_{i =0}^{j'} d_i + v' \sum_{i =0}^j d_i \implies (v - v') \cdot \sum_{i = j'+1}^j d_i > 0\,,
	\end{align*}
	 which is a contradiction because $v' > v$, $j > j'$ and $d_i \geq 0$ for all $0 \leq i \leq K$. Therefore, $s^*(\cdot)$ is non-decreasing.
	 
	 Next, we establish the left-continuity of $s^*(\cdot)$. Consider any sequence of increasing values $\{v_n\}_n$ such that $\lim_{n \to \infty} v_n = v$. Since $s^*(\cdot)$ is non-decreasing, there exists a $b_j \leq s^*(v)$ such that $\lim_{n \to \infty} s^*(v_n) = b_j$. Since there are only finitely many bids, there exists an $N \in \mathbb N$ such that $s^*(v_n) = b_j$ for all $n \geq N$. Therefore, for any $\ell > j$, we have
	 \begin{align*}
	 	(v_n - b_j) \cdot \sum_{i=0}^j d_i \geq (v_n - b_\ell) \cdot \sum_{i=0}^\ell d_i
	 \end{align*}
	 for all $n \geq N$. Taking the limit $n \to \infty$ on both sides yields
	 \begin{align*}
	 	(v - b_j) \cdot \sum_{i=0}^j d_i \geq (v - b_\ell) \cdot \sum_{i=0}^\ell d_i \quad \forall\ \ell > j\,.
	 \end{align*}
	 Therefore, $s^*(v) \leq b_j$. Combining this with $b_j \leq s^*(v)$ yields $b^*(v) = b_j$, thereby establishing the left-continuity of $s^*(\cdot)$.
\end{proof}

\subsection{Proof of Theorem~\ref{thm:convex-refor}}

\begin{proof}
	We have already established (1) and (2) in Section~\ref{sec:convex}, and only need to prove (3). First, observe that $F$ is absolutely continuous with $F(0) = 0$ and $F(1) = 1$.  Consequently, the Intermediate Value Theorem implies $\text{range}(F) = [0,1]$. Therefore, part (4) of Proposition 1 of \citet{embrechts2013note} implies $F(F^-(u)) = u$ for all $u \in [0,1]$. Consequently,
	\begin{align*}
		\mathbb{P}(s(v) = b_i) = F\left( F^-\left(1 - p_{i} \right) \right) - F\left(F^-\left(1 - p_{i} \right)  \right) = p_i - p_{i+1}\,.
	\end{align*} 
	Hence, $\mathbb{P}(s(v) \geq b_i) = p_i$ and part (2) applies. Consequently, $u(s|F, \pmb d) = u(\pmb p| F, \pmb d)$.
\end{proof}

%% file: appendix_known.tex
\section{Proofs for Section~\ref{sec:known_dist}}\label{appendix:known_GA}

The following lemma characterizes the update step of Algorithm~\ref{alg:known_GA}. It plays a vital role in our analysis of Algorithm~\ref{alg:known_GA}. Intuitively, projecting onto $\pp$ involves a modification of isotonic-regression which ensures the `no over-bidding' condition by ensuring $F^-(1- p_j) \geq b_j$. Someone versed in the Pool Adjacent Violators Algorithm (PAVA) for isotonic regression will find the characterization of the projection and the analysis familiar.

\begin{lemma}\label{lemma:GA_update}
	Fix bidding probabilities $\pmb p \in \pp$, step size $\eta > 0$, and highest competing bid $h = b_i$. Define
	\begin{align*}
		\pmb p' \coloneqq \argmin_{\pmb q \in \pp} \|\pmb q - \pmb p^+\| \quad \text{where} \quad \pmb p^+ \coloneqq \pmb p + \eta \cdot \grad u(\pmb p|F, b_i)\,.
	\end{align*}
	
	Moreover, let $\ell \coloneqq \min\{j > i \mid p_j \leq \eta \cdot \epsilon\}$ and
	\begin{align*}
		m \coloneqq \min_j\left\{ j \leq i\ \biggr|\ p_j \leq 1 - F(b_i) \text{ and } \sum_{k=j}^i (p_j - p_k) \leq \eta \cdot (F^-(1 - p_i) - b_i) \right\}\,.
	\end{align*}
	Then, we have 
    \begin{align}\label{eq:projection}
        p'_j &= \begin{cases}
            p_j &\text{if } j < m\\
            x &\text{if } m \leq j \leq i\\
            p_j - \eta \cdot \epsilon &\text{if } i < j < \ell\\ 
            0 &\text{if } j \geq \ell
        \end{cases}
    \end{align}
    where $x = \min\left\{ \frac{\eta \cdot(F^{-}(1 - p_i) - b_i) + \sum_{k=m}^{i} p_k}{i - m + 1}, 1 - F(b_i) \right\}$. Moreover, $x \geq p_j$ for all $j \in [m, i]$.
\end{lemma}

\begin{proof}
	For the purposes of the proof, define $\pmb p'$ using \eqref{eq:projection}. To establish the lemma, we need to show that $\pmb p'$ is a solution to the following optimization problem:
	\begin{align}\label{eq:proj_opt}
		\min \quad & \frac{1}{2} \|\p^+ - \pmb q\|^2\\
		\text{s.t.} \quad &\pmb q \cdot \e_j = q_j \leq 1 - F(b_j) && j \in [K] \nonumber\\
		&\pmb q \cdot (-\e_j) = -q_j \leq 0 && j \in [K] \nonumber \\
		&\pmb q \cdot (\e_{j+1} - \e_j)= q_{j+1} - q_{j} \leq 0 && j \in [K-1] \nonumber
	\end{align}
	
	By the KKT optimality conditions, it suffices to show that $\p^+ - \p'$, which is the gradient of the objective function at $\pmb q = \p'$, lies in the cone formed by the coefficient vectors of the tight constraints. We establish this fact in two mutually exclusive and exhaustive cases based on the value of $x$.
	\begin{itemize}
		\item CASE I: Assume $x = \frac{\eta \cdot(F^{-}(1 - p_i) - b_i) + \sum_{k=m}^{i} p_k}{m- i + 1}$. In this case, $\pmb p' \in \pp$ satisfies
			\begin{align*}
				\p' \cdot (\e_{j+1} - \e_j) = 0 && j \in [m, i-1]\\
				\pmb q \cdot (-\e_j) = 0 && j  \in [\ell, K]
			\end{align*}
		 i.e., the constraints $\pmb q \cdot (\e_{j+1} - \e_j) \leq 0$ for $j \in [m, i-1]$ and $\pmb q \cdot (-\e_j) \leq 0$ for $j \in [\ell, K]$ are tight at $\pmb q = \p'$. Define dual variables $\lambda_j$ as follows:
		 \begin{align*}
		 	\lambda_j = \begin{cases}
		 		\eta \cdot (F^{-}(1 - p_i) - b_i) + \sum_{k=j+1}^{i} (p_k - x) &\text{if } j \in [m, i-1]\\
            	-p^+_j &\text{if } j \in [\ell, K]
		 	\end{cases}
		 \end{align*}
		 The condition $x = \frac{\eta \cdot(F^{-}(1 - p_i) - b_i) + \sum_{k=m}^{i} p_k}{m- i + 1}$ implies
		 \begin{align}\label{eq:x_refor}
		 	\eta \cdot (F^{-}(1 - p_i) - b_i) + \sum_{k=m}^{i} (p_k - x) = \eta \cdot (F^{-}(1 - p_i) - b_i) + \sum_{k=m}^{i} p_k - (i - m +1) \cdot x = 0\,.
		 \end{align}
		 
		Therefore, we can write
   	\begin{align*}
   		\p^+ - \p' &= \sum_{j=m}^{i-1} (p_j - x) \cdot \e_j + \left(p_i + \eta \cdot (F^{-}(1 - p_i) - b_i) - x \right) \cdot \e_i + \sum_{j=\ell}^K p_j^+ \cdot \e_j\\
   		&= (0 - \lambda_m) \cdot \e_m + \sum_{j=m+1}^{i-1} (\lambda_{j-1} - \lambda_j) \cdot \e_{j} + \lambda_{i-1} \cdot \e_i + \sum_{j=\ell}^K (-p_j^+) \cdot (-\e_j)\\
   		&= \sum_{j=m}^{i-1} \lambda_j \cdot (\e_{j+1} - \e_j) + \sum_{j=\ell}^K \lambda_j \cdot (-\e_j)
   	\end{align*}
   	
   		To establish the KKT conditions for the current case, all that remains to be shown is $\lambda_j \geq 0$. This is trivially true for $j \in [\ell, K]$ because $\lambda_j = - p_j^+$ and the definition of $\ell$ implies that $p_j^+ \leq 0$ for all $j \in [\ell, K]$. To prove it for $j \in [m, i-1]$, note that the definition of $m$ implies $\sum_{k=m}^i (p_j - p_k) \leq \eta \cdot (F^-(1 -p_i) - b_i)$, which in turn implies
   		\begin{align*}
   			(i - m + 1) \cdot p_{j+1} &\leq \sum_{k=m}^i p_k + \eta \cdot (F^-(1 -p_i) - b_i)= (i - m+1) \cdot x\,.
   		\end{align*}
   		Therefore, $p_m \leq x$. Since $p_j \leq p_m$ for all $j \geq m$, we get that $p_j \leq x$ for all $j \in [m, i]$, and consequently $\lambda_m \leq \lambda_{m+1} \leq \dots \leq \lambda_{i-1}$. Finally, note that \eqref{eq:x_refor} implies
   		\begin{align*}
   			\lambda_m \geq \lambda_m + (p_m - x) = \eta \cdot (F^{-}(1 - p_i) - b_i) + \sum_{k=m}^{i} (p_k - x) = 0\,,
   		\end{align*}
   		thereby establishing $\lambda_j \geq 0$ for all $j \in [m, i-1]$.

	\item CASE II: Assume $x = 1 - F(b_i)$. In this case, $\pmb p' \in \pp$ satisfies
			\begin{align*}
				\p' \cdot \e_j = 1 - F(b_i)\\
				\p' \cdot (\e_{j+1} - \e_j) = 0 && j \in [m, i-1]\\
				\pmb q \cdot (-\e_j) = 0 && j  \in [\ell, K]
			\end{align*}
		 i.e., the constraints $\pmb q \cdot (\e_{j+1} - \e_j) \leq 0$ for $j \in [m, i-1]$, $\pmb q \cdot (-\e_j) \leq 0$ for $j \in [\ell, K]$ and $\pmb q \cdot \e_i = 1 - F(b_i)$ are tight at $\pmb q = \p'$. Define dual variables $\lambda_j$ as follows:
		 \begin{align*}
		 	\lambda_j = \begin{cases}
		 		\sum_{k=m}^j (x - p_j)  &\text{if } j \in [m, i-1]\\
		 		\eta \cdot(F^{-}(1 - p_i) - b_i) + \left\{\sum_{k=m}^{i} p_i \right\}  -  (i-m+1) \cdot x &\text{if } j = i\\
            	-p^+_j &\text{if } j \in [\ell, K]
		 	\end{cases}
		 \end{align*}
		 The definition of $m$ implies $p_j \leq 1 - F(b_i) =x$ for all $j \in [m, i-1]$, i.e, $\lambda_j \geq 0$ for all $j \in [m,i-1]$. Moreover, the definition of $x$ along with the condition $x = 1 - F(b_i)$ implies $\lambda_i \geq 0$. As before, we have $\lambda_j \geq 0$ for all $j \in [\ell, K]$ because $p_j^+ \leq 0$ for all $j \in [\ell, K]$ by definition of $\ell$. To establish the KKT conditions, note that
   	\begin{align*}
   		\p^+ - \p' &= \sum_{j=m}^{i-1} (p_j - x) \cdot \e_j + \left(p_i + \eta \cdot (F^{-}(1 - p_i) - b_i) - x \right) \cdot \e_i + \sum_{j=\ell}^K p_j^+ \cdot \e_j\\
   		&= (0 - \lambda_m) \cdot \e_m + \sum_{j=m+1}^{i-1} (\lambda_{j-1} - \lambda_j) \cdot \e_{j} + (\lambda_{i-1} + \lambda_i) \cdot \e_i + \sum_{j=\ell}^K (-p_j^+) \cdot (-\e_j)\\
   		&= \sum_{j=m}^{i-1} \lambda_j \cdot (\e_{j+1} - \e_j) + \lambda_i \cdot \e_i + \sum_{j=\ell}^K \lambda_j \cdot (-\e_j)\,.
   	\end{align*}
	\end{itemize}
	In both cases, we have shown that $\p^+ - \p'$ lies in the cone formed by the coefficient vectors of the tight constraints. Therefore, by KKT Theorem, $\p'$ is an optimal solution for the quadratic optimization problem \eqref{eq:proj_opt}. Moreover, in both cases we established that $x \geq p_j$ for all $j \in [m, i]$, thereby concluding the proof.
\end{proof}

\subsection{Proof of Proposition~\ref{prop:sub-linear-regret}}

\begin{proof}
	From the regret analysis of Online Gradient Descent (e.g., see Theorem~5.3.1. of \citealt{hazan2016introduction}), we get that
	\begin{align*}
		\sum_{t=1}^T \E[u(\p| F, h_t)] - \sum_{t=1}^T \E[u(\p_t| F, h_t)] \leq \frac{\|\p\|^2}{2 \eta} + \eta \cdot \sum_{t=1}^T \|\nabla u(\p_t| F, h_t)\|^2 &&\forall \p \in \pp\,.
	\end{align*}
	
	Since $\|\p\|^2 \leq K$ for all $\p \in \pp$ and $\|\nabla u(\p| F, h)\|^2 \leq \epsilon^2 \cdot (K-1) + 1 \leq 2$ for all $\p \in \pp, h \in \{b_0, \dots, b_K\}$, we have
	\begin{align}\label{eq:sublinear-regret-inter-1}
		\sum_{t=1}^T \E[u(\p| F, h_t)] - \sum_{t=1}^T \E[u(\p_t| F, h_t)] \leq \frac{K}{2 \eta} + 2\eta \cdot T &&\forall \p \in \pp\,.
	\end{align}
	
	Next, conditioning on $h_t$ and the past highest competing bids $\{h_s\}_{s=1}^{t-1}$, Theorem~\ref{thm:convex-refor} implies 
	\begin{align*}
		u(A_t| F, h_t) = u(\p_t | F, h_t) \quad \text{and} \quad u(s_\p| F, h_t) = u(\p | F, h_t)\,,
	\end{align*}
    where $s(\cdot)$ is the bidding strategy corresponding to $\p$, i.e., 
    \begin{align*}
				s_\p(v) = b_i \quad \text{for } v \in \left(F^-\left(1- p_i \right), F^-\left(1 - p_{i+1} \right) \right]\,,
	\end{align*}
    and $s_\p(0) = 0$. Here, we have used the fact that $V_t$ is independent of $h_t$. Taking expectation over $h_t$ and $\{h_s\}_{s=1}^{t-1}$, and summing over $t=1$ to $T$ yields
    \begin{align*}
        \sum_{t=1}^T \E[u(A_t| F, h_t)] = \sum_{t=1}^T \E[u(\p_t| F, h_t)] \quad \text{and} \quad \sum_{t=1}^T \E[u(s| F, h_t)] = \sum_{t=1}^T \E[u(\p| F, h_t)]\,.
    \end{align*}

	Combining with \eqref{eq:sublinear-regret-inter-1} yields
	\begin{align*}
		\sum_{t=1}^T \E[u(s_\p| F, h_t)] - \sum_{t=1}^T \E[u(A_t| F, h_t)] \leq \frac{K}{2 \eta} + 2\eta \cdot T &&\forall \p \in \pp\,.
	\end{align*}
    Finally, recall that Lemma~\ref{lemma:non-decreasing-opt} establishes the optimality of non-decreasing left-continuous strategy that never overbid, all of which can be written as $s_\p$ for some $\p$ due to Theorem~\ref{thm:convex-refor}. This implies our desired bound:
    \begin{align*}
		\max_{s(\cdot)} \sum_{t=1}^T \E[u(s| F, h_t)] - \sum_{t=1}^T \E[u(A_t| F, h_t)] \leq \frac{K}{2 \eta} + 2\eta \cdot T\, \tag*{\qedhere}
	\end{align*}
\end{proof}

\subsection{Proof of Proposition~\ref{prop:regret-lower-bound}}

\begin{proof}
	Let the value distribution $F$ be the uniform distribution on the interval $[1/2, 1/2 + 1/T]$. Moreover, let the set of possible bids be $b_0 = 0$ and $b_1 = 1/4$. In each auction $t \in [T]$, suppose the highest competing bid $h_t$ is set equal to 0 and $1/4$ with equal probability, i.e., $\mathbb{P}(h_t = 0) = \mathbb{P}(h_t = 1/4) = 1/2$, and assume that these highest competing bids $\{h_1, \dots, h_T\}$ are independent across auctions and independent of the values $\{V_t\}_{t=1}^T$.
	
	For any auction $t \in [T]$ and a bidding strategy $A_t(\cdot)$ that does not depend on the realization of $h_t$, it is clear that
	\begin{align*}
		\E_{h_t}[u(A_t| F, h_t)] &= \E_{h_t}\left[ \E_{V_t}[ (V_t - A_t(V_t)) \cdot \mathbf 1(A_t(V_t) \geq h_t)] \right]\\
		&= \E_{V_t}\left[ \E_{h_t}[ (V_t - A_t(V_t)) \cdot \mathbf 1(A_t(V_t) \geq h_t)] \right]\\
		&\leq \E_{V_t}\left[\max_{b \in \{0, 1/4\}} \E_{h_t}[ (V_t - b) \cdot \mathbf 1(b \geq h_t)] \right]\\
		&= \E_{V_t}\left[\max_{b \in \{0, 1/4\}} (V_t - b) \cdot \mathbb{P}(b \geq h_t) \right]\\
		&\leq \frac{1}{4} + \frac{1}{T}
	\end{align*}
	where the last inequality follows from the fact that, when $h_t$ is selected uniformly at random from $\{0,1/4\}$, bidding 1/4 yields higher utility than bidding $0$ for all $V_t \in [1/2, 1/2 +1/T]$. This is because $V_t - 1/4 \geq (V_t - 0)/2$ for all $V_t \in [1/2, 1/2 +1/T]$. Therefore, we get an upper bound on the expected performance of every online algorithm
	\begin{align}\label{eq:lower-bound-inter-1}
		\E_{h_1, \dots, h_T} \left[ \sum_{t=1}^T u(A_t| F, h_t) \right] \leq \left( \frac{1}{4} + \frac{1}{T}\right) \cdot T = \frac{T}{4} + 1\,.
	\end{align}
	
	On the other hand, note that the anti-concentration of sums of independent Bernoulli random variables implies the existence of a constant $c > 0$ such that
	\begin{align*}
		\mathbb{P}\left(E \right) \geq c\, \quad \text{where} \quad E = \left\{ \sum_{t=1}^T \mathbf 1(h_t = 0) \geq \frac{T}{2} + \frac{\sqrt{T}}{2} \right\}\,.
	\end{align*}
	
	Let $s_0$ (respectively $s_{1/4}$) be the bidding strategy that always bids $0$ ($1/4$), i.e., $s_0(v) = 0$ and $s_{1/4}(v) = 1/4$ for all $v \in [0,1]$. Then, we get
	\begin{align}\label{eq:lower-bound-inter-2}
		\E_{h_1,\dots, h_T} \left[ \max_{s(\cdot)} \sum_{t=1}^T u(s| F, h_t) \right] &\geq \E_{\{h_t\}_t} \left[ \sum_{t=1}^T u(s_0| F, h_t)\ \biggr|\ E \right] \cdot c + \E_{\{h_t\}_t} \left[ \sum_{t=1}^T u(s_{1/4}| F, h_t)\ \biggr|\ E^c \right] \cdot (1 - c) \nonumber\\
		&\geq \frac{1}{2}\cdot \left(\frac{T}{2} + \frac{\sqrt{T}}{2} \right) \cdot c + \frac{1}{4} \cdot T \cdot (1 - c) \nonumber\\
		&\geq \frac{T}{4} + \frac{c}{2} \cdot \sqrt{T}\,.
	\end{align}
	Combining \eqref{eq:lower-bound-inter-1} and \eqref{eq:lower-bound-inter-2} yields
	\begin{align*}
		\E_{h_1,\dots, h_T} \left[ \regret(A|F) \right] = \E_{h_1,\dots, h_T} \left[ \max_{s(\cdot)} \sum_{t=1}^T u(s| F, h_t) - \sum_{t=1}^T u(A_t| F, h_t)\right] \geq \frac{c}{2} \cdot \sqrt{T}\,.
	\end{align*}
	Therefore, we must have $\max_{h_1, \dots, h_T} \regret(A|F) \geq \Omega(\sqrt{T})$ for all online algorithms $A$, thereby establishing the lemma.
\end{proof}

\subsection{Proof of Theorem~\ref{thm:strat_robust}}

\begin{proof}
	We will use a potential function argument: we define a function $\Phi: \pp \to [0, \sqrt{KT/2}]$ such that for all auctions $t \in [T]$, the difference between the revenue of the seller $\rev(A_t, h_t)$ and $\mye(F)$ can be charged against the change in $\Phi$ for all possible values of $h_t$. To this end, define the potential function $\Phi$ as
	\begin{align*}
		\Phi(\p) \coloneqq \frac{\|\p\|_2^2}{2\eta} = \sum_{j=1}^K \frac{p_j^2}{2\eta}\,.
	\end{align*}
    Note that $0 \leq \Phi(\p) \leq K/2\eta = \eta T$ for all $\p \in \pp$.
 
	We start by showing that, to establish the theorem, it suffices to prove the following statement for all $\p \in \pp$ and $h \in \{b_0, \dots, b_K\}$:
	\begin{align}\label{eq:strat_robust-inter-1}
		\Delta \Phi (\p) + \rev(\p, h) \leq \mye(F) + \eta\,,
	\end{align}
	where 
	\begin{itemize}
		\item $\Delta \Phi(\p) = \Phi\left(\argmin_{\p' \in \pp} \left\|\p' - \left\{\p + \eta\cdot \nabla u(\p|F, h_t) \right\} \right\| \right) - \Phi(\p)$ is the change in potential caused by a single update-step of Algorithm~\ref{alg:known_GA},
		\item $\rev(\p, h) \coloneqq \rev(s, h)$ for the bidding strategy $s$ corresponding to bidding-probability vector $\p$, which sets $s(v) \coloneqq b_k$ for $v \in \left(F^-\left(1- p_i \right), F^-\left(1 - p_{i+1} \right) \right]$.
	\end{itemize}
	This is because, conditioned on $h_t$ and the past competing bids $\{h_s\}_{s=1}^{t-1}$, applying \eqref{eq:strat_robust-inter-1} to iterate $\p_t$ of Algorithm~\ref{alg:known_GA} yields
	\begin{align*}
		\Phi(\p_{t+1}) - \Phi(\p_t) + \rev(\p_t, h_t) \leq \mye(F) + \eta\,.
	\end{align*}
    Here, we have used the fact that $V_t$ is independent of $\{h_s\}_{s=1}^t$, which ensures that the distribution of $V_t$ remains $F$ even after the conditioning. Taking expectations over $h_t$ and $\{h_s\}_{s=1}^{t-1}$ yields
    \begin{align*}
        \E[\Phi(\p_{t+1})] - \E[\Phi(\p_t)] + \E[\rev(\p_t, h_t)] \leq \mye(F) + \eta\,.
    \end{align*}
	Summing over all times steps and noting that $\rev(\p_t, h_t) = \rev(A_t, h_t)$ for all $h_t$, we get
	\begin{align*}
		&\sum_{t=1}^T \E[\Phi(\p_{t+1})] - \E[\Phi(\p_t)] + \sum_{t=1}^T \E[\rev(A_t, h_t)] \leq \mye(F) \cdot T + \eta T\\
		\implies &\sum_{t=1}^T \E[\rev(A_t, h_t)] \leq \mye(F) + \eta T + \Phi(\p_1) - \E[\Phi(\p_{T+1})] \leq  \mye(F) + \eta T + \eta T\,,
	\end{align*}
	and the theorem statement follows. To complete the proof, we next establish \eqref{eq:strat_robust-inter-1}.
	
	Fix some $\p \in \pp$ and highest competing bid $h = b_k$ for some $k \in \{0, \dots, K\}$. First, observe that for $s(v) \coloneqq b_i$ for $v \in \left(F^-\left(1- p_i \right), F^-\left(1 - p_{i+1} \right) \right]$, we have
	\begin{align*}
		\rev(\p,h) = \rev(s, h) &= \E\left[ s(v) \cdot \mathbf 1 (s(v) \geq b_k) \right]\\
		&= \sum_{i=k}^K b_i \cdot \mathbb P(s(v) = b_i)\\
		&= \sum_{i=k}^K b_i \cdot (p_i - p_{i+1})\\
		&= b_k \cdot p_k +  \sum_{i=k+1}^K p_i \cdot (b_i - b_{i-1})\\
		&= b_k \cdot p_k + \sum_{i=k+1}^K p_{i} \cdot \epsilon \tag{$\clubsuit$}\label{eq:strat_robust-inter-2}
	\end{align*}
	
	Next, using the terminology and the result from Lemma~\ref{lemma:GA_update}, we can write
	\begin{align*}
		\Delta \Phi(\p) = & \Phi(\p') - \Phi(\p)\\
		= & \sum_{j=1}^K \frac{(p'_j)^2 - p_j^2}{2\eta}\\
		= & \frac{1}{2\eta} \cdot \sum_{j=1}^K (p'_j - p_j)(p_j' + p_j)\\
		= & \frac{1}{2\eta} \cdot \sum_{j=1}^{m-1} (p_j - p_j)(p_j + p_j) + \frac{1}{2\eta} \cdot \sum_{j=m}^k (x - p_j)(x+p_j)\\
		 & + \frac{1}{2\eta} \cdot \sum_{j=k+1}^{\ell-1} (p_j - \eta \epsilon - p_j)(p_j - \eta \epsilon + p_j) + \frac{1}{2\eta} \cdot\sum_{j=\ell}^K (0-p_j)(0+p_j) \\
		 \leq & \frac{1}{2\eta} \cdot \sum_{j=m}^k (x-p_j)(2x) + \frac{1}{2\eta} \cdot \sum_{j=k+1}^{\ell-1} (-\eta \epsilon)(2p_j - \eta \epsilon)\\
		 \leq & \frac{x}{\eta} \cdot \sum_{j=m}^k (x-p_j) - \sum_{j=k+1}^{\ell-1} p_j \cdot \epsilon + \sum_{j=k+1}^{\ell-1}\frac{\eta \epsilon^2}{2}\\
		 \leq & \frac{x}{\eta} \cdot \left\{ \eta \cdot (F^-(1-p_k) - b_k) \right\} - \sum_{j=k+1}^{\ell-1} p_j \cdot \epsilon + \frac{(\ell-k-2)\eta \epsilon^2}{2}\\
		 = & p_k \cdot (F^-(1-p_k) - b_k) + (x-p_k) \cdot (F^-(1-p_k) - b_k) - \sum_{j=k+1}^{\ell-1} p_j \cdot \epsilon + \frac{(\ell-k-2)\eta \epsilon^2}{2}\\
		 \leq & p_k \cdot (F^-(1-p_k) - b_k) + \eta \cdot (1- \epsilon) - \sum_{j=k+1}^{\ell-1} p_j \cdot \epsilon + \frac{(\ell-k-2)\eta \epsilon^2}{2} \tag{$\spadesuit$}\label{eq:strat_robust-inter-3}
	\end{align*}
	where the third inequality follows from the definition of $x$ (as defined in Lemma~\ref{lemma:GA_update}), and the final inequality follows from the fact that $x - p_k \leq \eta$ and $F^-(1-p_j) - b_j \leq 1- \epsilon$ for all $j \in \{0,1,\dots, K\}$.
	
	Combining \eqref{eq:strat_robust-inter-2} and \eqref{eq:strat_robust-inter-3} yields
	\begin{align*}
		\Delta \Phi(\p) + \rev(\p, h) &\leq p_k \cdot (F^-(1-p_k) - b_k) + \eta \cdot (1- \epsilon) - \sum_{j=k+1}^{\ell-1} p_j \cdot \epsilon + \frac{(\ell-k-2)\eta \epsilon^2}{2} + b_k \cdot p_k + \sum_{i=k+1}^K p_{i} \cdot \epsilon\\
		&\leq \mye(F) + \eta \cdot (1-\epsilon) + \frac{(\ell-k-2)\eta \epsilon^2}{2} + \sum_{j=\ell}^K p_j \cdot \epsilon\\
		&\leq \mye(F) + \eta \cdot (1-\epsilon)+ (\ell-k-2)\eta \epsilon^2 + \sum_{j=\ell}^K \eta \epsilon^2\\
		&\leq \mye(F) + \eta \cdot (1-\epsilon)+ \eta \epsilon^2 \cdot K\\
		&\leq \mye(F) + \eta \cdot (1-\epsilon)+ \eta \epsilon\\
		&= \mye(F) + \eta\,,
	\end{align*} 
	where the second inequality follows from $\mye(F) \geq r \cdot (1 - F(r))$ for $r = F^-(1 - p_i)$, the third inequality follows from the fact that $p_j \leq \eta\epsilon$ for all $j \geq \ell$ ((see Lemma~\ref{lemma:GA_update} for definition of $\ell$)), and the fifth inequality follows from the assumption that $b_K = \epsilon K \leq 1$. Thus, we have established \eqref{eq:strat_robust-inter-1} and thereby the theorem.
\end{proof}

%% file: appendix_unknown_dist.tex
\section{Proofs for Section~\ref{sec:unknown_dist}}\label{appendix:unknown}

\subsection{Proof of Proposition~\ref{prop:alg_equivalence}}

\begin{proof}
	We will prove the proposition using induction on $t \in [T]$. The base case $t=1$ follows from our assumption that $\v_1 = \pmb 1 - \p_1$. Assume that the induction hypothesis holds for some $t \in [T-1]$, i.e., $\v_t = \pmb 1 - \p_t$. Suppose the $t$-th highest competing bid $h_t$ is equal to $b_k$ for some $k \in [K]$. Observe that
	\begin{align*}
		&\text{For } i > k: \quad v_{t+1, i}^+ = v_{t,i} + \eta \cdot \epsilon = 1 - p_{t,i}^+  + \eta \cdot \epsilon = 1 - (p_{t,i} - \eta \cdot \epsilon) = 1 - p_{t+1, i}\\
		&\text{For } i = k: \quad v_{t+1, i}^+ = v_{t,i} - \eta \cdot (v_{t,i} - b_k) = 1 - p_{t,i} - \eta \cdot (F^-(1 - p_{t,i}) - b_k) = 1- p_{t+1,i}\\
		&\text{For } i < k: \quad v_{t+1, i}^+ = v_{t,i} = 1 - p_{t,i} = 1 - p_{t+1,i}\,.
	\end{align*}
	Therefore we have $\v^+_{t+1} = \pmb 1 - \p^+_{t+1}$. Next, note that $F(x) = x$ implies
	\begin{align*}
		\vv &= \{\v \in [0,1]^K \mid v_i \leq v_{i+1}, v_i \geq b_i\}\\
		 &= \{\v \in [0,1]^K \mid 1 - v_i \geq 1 - v_{i+1}, 1 - v_i \leq 1 - F(b_i)\}\\
		 &= \{\pmb 1 - \p \in [0,1]^K \mid p_i \geq p_{i+1}, p_i \leq 1 - F(b_i)\} \tag{Setting $v_i = 1- p_i$}\\
		 &= \pmb 1 - \pp\,.
	\end{align*}
	Finally, combining $\v^+_{t+1} = \pmb 1 - \p^+_{t+1}$ and $\vv = \pmb 1 - \pp$, we get
	\begin{align*}
		\v_{t+1} = \argmin_{\v \in \vv} \|\v - \v_t^+\| = \argmin_{\v \in \vv} \|(\pmb 1 - \v) - (\pmb 1 - \v_t^+)\| = \argmin_{\p \in \pp} \|\p - \p_t^+\| =  \p_{t+1}\,.
	\end{align*}
	This completes the induction step.
\end{proof}

As a direct consequence of Proposition~\ref{prop:alg_equivalence} and Lemma~\ref{lemma:GA_update}, we get the following corollary characterizing the update step of Algorithm~\ref{alg:threshold}. We will use it repeatedly in our proofs for Algorithm~\ref{alg:threshold}.

\begin{corollary}\label{cor:threshold_update}
	Fix thresholds $\v \in \vv$, step size $\eta > 0$, and highest competing bid $h = b_i$. Define
	\begin{align*}
		\v' \coloneqq \argmin_{\pmb w \in \vv} \|\pmb w - \v^+\| \quad \text{where} \quad v_{j}^+ = \begin{cases}
			 	v_{j} + \eta \cdot \epsilon &\text{if } b_j > h\\
			 	v_{j} - \eta \cdot (v_{j} - h_t) &\text{if } b_j = h\\
			 	v_{j} &\text{if } b_j < h_t
			 \end{cases}\,.
	\end{align*}
	
	Moreover, let $\ell \coloneqq \min\{j > i \mid v_j \geq 1 - \eta \cdot \epsilon\}$ and
	\begin{align*}
		m \coloneqq \min_j\left\{ j \leq i\ \biggr|\ v_j \geq b_i \text{ and } \sum_{k=j}^i (v_k - v_j) \leq \eta \cdot (v_i - b_i) \right\}\,.
	\end{align*}
	Then, we have 
    \begin{align}\label{eq:threshold_projection}
        v'_j &= \begin{cases}
            v_j &\text{if } j < m\\
            x &\text{if } m \leq j \leq i\\
            v_j + \eta \cdot \epsilon &\text{if } i < j < \ell\\ 
            1 &\text{if } j \geq \ell
        \end{cases}
    \end{align}
    where $x = \max\left\{\frac{ \left\{\sum_{k=m}^{i} v_k \right\} -\eta \cdot(v_i - b_i) }{i - m + 1}, b_i \right\}$. Moreover, $x \leq v_j$ for all $j \in [m, i]$.
\end{corollary}

\subsection{Proof of Theorem~\ref{thm:regret_threshold}}

\begin{proof}
	Consider a benchmark bidding strategy $s^*: [0,1] \to \{b_0, b_1, \dots, b_K\}$ and a value $v^* \in [0,1]$. Assume $s^*(w) \leq w$ for all $w \in [0,1]$,  and let $s^*(v^*) = b_{i^*}$ be the bid for value $v$ under $s^*$. Define the potential function $\Phi: \vv \to [0, K/\eta]$ as
	\begin{align*}
		\Phi(\v|v^*) \coloneqq \frac{1}{\eta} \cdot \left\{ \sum_{j=1}^K (v^* - v_j) \cdot \mathbf 1 (v^* > v_j) + \sum_{j=1}^{i^*} v_j \right\}
	\end{align*}
	We start by showing that, to prove the theorem, it suffices to prove the following statement for all thresholds $\v \in \vv$, values $v^* \in [0,1]$, and highest competing bids $h$:
	\begin{align}\label{eq:regret_threshold_inter-1}
		\Delta \Phi(\v|v^*) + R(\v | s^*, v^*, h) \leq 3 \cdot \mathbf{1}\left(\min_{j \in [K]} |v_j - v^*| \leq \eta \right)
	\end{align}
	where
	\begin{itemize}
		\item $\Delta \Phi(\v|v^*) \coloneqq \Phi(\v'|v^*) - \Phi(\v|v^*)$ is the change in potential caused by a single update step of Algorithm~\ref{alg:threshold}. Here $\v'$ are the thresholds obtained by applying the update step of Algorithm~\ref{alg:threshold} to $\v$ (see equation \eqref{eq:threshold_projection} of Corollary~\ref{cor:threshold_update} for a formal definition).
		\item $R(\v | s^*, v^*, h) \coloneqq \{v^* - s^*(v^*)\} \cdot \mathbf{1}(h \leq s^*(v^*)) - \{v^* - s_\v(v^*)\} \cdot \mathbf{1}(h \leq s_\v(v^*))$ is the regret associated with bidding according to $s_\v$ instead of the benchmark strategy $s^*$ for a buyer with value $v^*$. Here $s_\v$ is the bidding strategy corresponding to thresholds $\v$, i.e., $s_\v(v^*) = b_j$ if $v^* \in (v_j, v_{j+1}]$.
	\end{itemize}
 
	Suppose \eqref{eq:regret_threshold_inter-1} holds for all thresholds $\v \in \vv$, values $v^* \in [0,1]$, and highest competing bids $h$. Then, conditioned on $h_t$ and the past competing bids $\{h_s\}_{s=1}^{t-1}$, we can apply it to $\v_t$ and $v^* = V_t$ to get
	\begin{align*}
		\Phi(\v_{t+1}|V_t) - \Phi(\v_t|V_t) + R(\v_t|s^*, V_t,h_t) \leq 3 \cdot \mathbf{1}\left(\min_{j \in [K]} |v_{t,j} - V_t| \leq \eta \right)\, &&\forall t \in [T]\,. 
	\end{align*}
    Taking an expectation over $V_t \sim F$, and using $\E_{V_t}[\cdot]$ to denote the conditional expectation $\E[\cdot | \{h_s\}_{s=1}^t]$, yields
    \begin{align*}
        &\E_{V_t}[\Phi(\v_{t+1}|V_t)] - \E_{V_t}[\Phi(\v_t|V_t)] + \E_{V_t}\left[ R(\v_t|s^*, V_t,h_t)  \right] \leq 3 \cdot \mathbb{P}_{V_t}\left(\min_{j \in [K]} |v_{t,j} - V_t| \leq \eta \right)\\
        \implies &\E_{V_t}[\Phi(\v_{t+1}|V_t)] - \E_{V_t}[\Phi(\v_t|V_t)] + u(s^*|F,h_t) - u(s_\v|F,h_t) \leq 3 \cdot \sum_{j=1}^K \mathbb{P}_{V_t}\left( |v_{t,j} - V_t| \leq \eta \right)\\
        \implies & u(s^*|F,h_t) - u(s_\v|F,h_t) \leq 3 \cdot K \cdot \bar f \cdot 2\eta + \E_{V_t}[\Phi(\v_{t+1}|V_t)] - \E_{V_t}[\Phi(\v_t|V_t)]\\
        \implies & u(s^*|F,h_t) - u(s_\v|F,h_t) \leq 3 \cdot K \cdot \bar f \cdot 2\eta + \E_{V_{t+1}}[\Phi(\v_{t+1}|V_{t+1})] - \E_{V_t}[\Phi(\v_t|V_t)]
    \end{align*}
    where the first implication follows from the definition of $u(s|F,h)$ and the union bound, the second implication follows from $\mathbb{P}\left(|v_{t,j} - v^*| \leq \eta \right) = \mathbb{P}\left(v^* \in  [v_{t,j} - \eta, v_{t,j} + \eta] \right) \leq \bar f \cdot 2\eta$ (here $\bar f$ is an upper bound on the density of $F$), and the third implication follows from the independence of $V_t$, $V_{t+1}$ and $\{h_s\}_{s=1}^t$. 
    
	Taking expectation over $\{h_s\}_{s=1}^T$ and summing over $t \in [T]$ yields:
	\begin{align*}
		 \sum_{t=1}^T \E[u(s^*|F,h_t)] - \sum_{t=1}^T \E[u(s_\v|F,h_t)] &\leq 6\bar f K \cdot \eta T + \sum_{t=1}^T \left\{\E[\Phi(\v_{t}|V_t)] - \E[\Phi(\v_{t+1}|V_{t+1}) ] \right\} \\
        &\leq 6\bar f K \cdot \eta T + \E[\Phi(\v_{1}|V_1)] - \E[\Phi(\v_{T+1}|V_{T+1}) ]\\
        &\leq 6\bar f K \cdot \eta T + \frac{K}{\eta}\,,
	\end{align*}
	where $V_{T+1}$ is a fresh sample from $F$, independent of all other values. Since the benchmark strategy $s^*$ was an arbitrary strategy that did not overbid, we have shown that \eqref{eq:regret_threshold_inter-1} is a sufficient condition for the theorem to hold, and we focus on establishing \eqref{eq:regret_threshold_inter-1} in the remainder of the proof.
	
	Fix a benchmark strategy $s^*$ with $s^*(w) \leq w$ for all $w \in [0,1]$, a value $v^*$, a highest competing bid $h$ and thresholds $\v \in \vv$. First assume that $\min_{j \in [K]} |v_j - v^*| > \eta$. In particular, this implies that if $v_j > v^*$ (respectively $v_j < v^*$), then $v_j'>v^*$ (respectively $v_j' < v^*$), i.e., the thresholds don't cross $v^*$ during the update $\v \to \v'$. Let $s^*(v^*) = b_{i^*}$ be the bid under strategy $s^*$ for value $v^*$, let $s_\v(v^*) = b_u$ be the bid under strategy $s_\v$ for value $v^*$ (i.e., $u = \max\{j \mid v_j \leq v^* \}$), and let $h = b_i$ be the highest competing bid. Since the thresholds don't cross $v^*$ during the update $\v \to \v'$, we have
 \begin{align*}
     \Delta \Phi(\v|v^*) =&  \frac{1}{\eta} \cdot \left\{ \sum_{j=1}^u (v^* - v'_j) \cdot \mathbf 1 (v^* > v'_j) + \sum_{j=1}^{i^*} v'_j \right\} - \frac{1}{\eta} \cdot \left\{ \sum_{j=1}^u (v^* - v_j) \cdot \mathbf 1 (v^* > v_j) + \sum_{j=1}^{i^*} v_j \right\}\\
     =& \frac{1}{\eta} \cdot \left\{ - \sum_{j=1}^u (v'_j - v_j) + \sum_{j=1}^{i^*} (v'_j - v_j) \right\}\,.
 \end{align*}
 We establish \eqref{eq:regret_threshold_inter-1} by separately analyzing the following mutually exclusive and exhaustive cases on the ordering of the algorithm's bid $b_{u}$, the benchmark bid $b_{i^*}$, and the highest other bid $b_{i}$. Throughout these cases, we extensively use our characterization in Corollary \ref{cor:threshold_update} of the form of the update $v'$.
	\begin{enumerate}
    \item $b_u \geq b_{i^*} \geq b_i$: The utility under $s^*$ is $\{v^* - s^*(v^*)\}\cdot \mathbf 1(h \leq s^*(v^*)) = v^* - b_{i^*}$ and the utility under $s_\v$ is $\{v^* - s_\v(v^*)\}\cdot \mathbf 1(h \leq s_\v(v^*)) = v^* - b_u$. Therefore, the regret $R(\v|s^*,v^*,h) = b_u - b_{i^*} = (u-i^*) \epsilon$. On the other hand, since $v_j' = v_j + \eta \epsilon $ for all $i^* < j \leq u$, we get $\Delta \Phi(\v|v^*) \leq - (u - i^*) \epsilon$.

    \item $b_{i^*} > b_u \geq b_i$: The utility under $s^*$ is $\{v^* - s^*(v^*)\}\cdot \mathbf 1(h \leq s^*(v^*)) = v^* - b_{i^*}$ and the utility under $s_\v$ is $\{v^* - s_\v(v^*)\}\cdot \mathbf 1(h \leq s_\v(v^*)) = v^* - b_u$. Therefore, the regret $R(\v|s^*,v^*,h) = -(b_{i^*} - b_u)$. On the other hand, since $v_j' \leq v_j + \eta \epsilon$ for all $u < j \leq i^*$, we get $\Delta \Phi(\v|v^*) \leq (i^* - u)\epsilon = b_{i^*} - b_u$.
  
    \item $b_{i} > b_u \geq b_{i^*}$: The utility under $s^*$ is $\{v^* - s^*(v^*)\}\cdot \mathbf 1(h \leq s^*(v^*)) = 0$ and the utility under $s_\v$ is $\{v^* - s_\v(v^*)\}\cdot \mathbf 1(h \leq s_\v(v^*)) = 0$ because $h = b_i > b_u \geq b_{i^*}$. Hence we have $R(\v|s^*, v^*, h) = 0$. On the other hand, $\Delta \Phi(\v|v^*) = 0$ because $v_j' = v_j$ for all $j < m$ and thresholds cannot cross $v^*$ during the update $\v \to \v'$.

    \item $b_i > b_{i^*} > b_u$: The utility under $s^*$ is $\{v^* - s^*(v^*)\}\cdot \mathbf 1(h \leq s^*(v^*)) = 0$ and the utility under $s_\v$ is $\{v^* - s_\v(v^*)\}\cdot \mathbf 1(h \leq s_\v(v^*)) = 0$ because $h = b_i > b_{i^*} > b_u$. Hence we have $R(\v|s^*, v^*, h) = 0$. On the other hand, $\Delta \Phi(\v|v^*) \leq 0$ because $v_j' \leq v_j$ for all $j \leq i$.

    \item $b_u \geq b_{i} > b_{i^*}$: The utility under $s^*$ is $\{v^* - s^*(v^*)\}\cdot \mathbf 1(h \leq s^*(v^*)) = 0$ because $h = b_i > b_{i^*}$, and the utility under $s_\v$ is $\{v^* - s_\v(v^*)\}\cdot \mathbf 1(h \leq s_\v(v^*)) = v^* - b_u$. Therefore, the regret $R(\v|s^*,v^*,h) = - (v^* - b_u)$. On the other hand, $v_j' = x$ for $m \leq j \leq i$ and $v_j' = v_j + \eta \epsilon$ for $i < j \leq u$. Moreover, the definition of $x$ implies $\sum_{j=i^*+1}^{i} (v_j - v_j') \leq \sum_{j=m}^i (v_j - x) \leq \eta (v_i - b_i)$, and the definition of $u$ implies $v^* \geq v_u \geq v_i$. Thus, we get $\Delta \Phi(\v|v^*) \leq (v_i - b_i) - (u-i)\epsilon \leq v^* - b_u$.

    \item $b_{i^*} \geq b_i > b_u$: The utility under $s^*$ is $\{v^* - s^*(v^*)\}\cdot \mathbf 1(h \leq s^*(v^*)) = v^* - b_{i^*}$ and the utility under $s_\v$ is $\{v^* - s_\v(v^*)\}\cdot \mathbf 1(h \leq s_\v(v^*)) = 0$ because $h = b_i > b_u$. Therefore, the regret $R(\v|s^*,v^*,h) = v^* - b_{i^*}$. On the other hand, $v_j' \leq v_j + \eta \epsilon$ for all $j > i$ and $v_j' = x$ for all $m \leq j \leq i$. The definition of $x$, the fact that $b_i \leq b_{i^*} = s^*(v^*) \leq v^* < v_i$, and the assumption that the thresholds don't cross $v^*$ implies $\sum_{j=m}^i (v_j -x) = \eta (v_i - b_i)$. As a consequence, we get $\sum_{j=u+1}^i (v'_j - v_j) = \sum_{j=m}^i (x - v_j) = -\eta (v_i - b_i)$. Moreover, the definition of $u$ implies $v_i > v^* > v_u$. Thus, we have $\Delta \Phi(\v|v^*) \leq (i^* - i)\epsilon - (v_i - b_i) \leq -(v^* - b_{i^*})$.
	\end{enumerate}

In all of the six cases, we have established the desired bound $\Delta \Phi(\v|v^*) + R(\v|s^*, v^*, h) \leq 0$. 
	
	To complete the proof of \eqref{eq:regret_threshold_inter-1}, we now relax our assumption that $\min_{j \in [K]} |v_j - v^*| > \eta$ and consider the setting where $\min_{j \in [K]} |v_j - v^*| \leq \eta$. First, observe that Corollary~\ref{cor:threshold_update} implies
	\begin{align*}
         &v'_j - v_j = 0 &\text{if } j < m\\
         & \sum_{j=m}^i (v'_j - v_j)   \geq -\eta \\
         &v'_j - v_j   \leq \eta \cdot \epsilon &\text{if } i < j\,,
    \end{align*}
    which in turn implies $\sum_{j=1}^K |v_j'- v_j| \leq 2\eta$. Next, observe that the following functions are 1-Lipschitz:
    \begin{align*}
        v_j \mapsto (v^* - v_j) \cdot \mathbf{1}(v^* > v_j) + v_j \quad \text{ and } \quad v_j \mapsto  (v^* - v_j) \cdot \mathbf{1}(v^* > v_j)
    \end{align*}
    As a consequence, we get $\Delta \Phi(\v|v^*) \leq 2$ for all $\v \in \vv$. On the other hand, $R(\v|s^*, v^*, h) \leq 1$ for all $\v, s^*, v^*, h$. Combining the two, we get the desired bound $\Delta \Phi(\v|v^*) + R(\v|s^*, v^*, h) \leq 3$ in the case when $\min_{j \in [K]} |v_j - v^*| \leq \eta$. This establishes \eqref{eq:regret_threshold_inter-1} and completes the proof.
\end{proof}

\subsection{Proof of Theorem~\ref{thm:threshold_robust}}

\begin{proof}
	We will use a potential-function argument with the potential function $\Phi: \pp \to [0, K/\eta]$ defined as
	\begin{align*}
		\Phi(\v) \coloneqq \frac{1}{\eta} \cdot \sum_{j=1}^K \int_{v_i}^1 (1 - F(t)) \cdot dt\,.
	\end{align*}
	Note that $0 \leq \Phi(\v) \leq K/\eta$ for all $\v \in \vv$.
	
	We start by showing that, to establish the theorem, it suffices to prove the following statement for all thresholds $\v \in \vv$ and highest competing bids $h \in \{b_0, \dots, b_K\}$:
	\begin{align}\label{eq:threshold_robust-inter-1}
		\Delta \Phi (\v) + \rev(\v, h) \leq \mye(F) + \eta \bar f\,,
	\end{align}
	where 
	\begin{itemize}
		\item $\Delta \Phi(\v) = \Phi(\v') - \Phi(\v)$ is the change in potential caused by a single update-step of Algorithm~\ref{alg:threshold}. Here, $\v'$ are the thresholds obtained by applying the update step of Algorithm~\ref{alg:threshold} to $\v$ (see Corollary~\ref{cor:threshold_update} for a formal definition).
		\item $\rev(\v, h) \coloneqq \rev(s_\v, h)$ for the bidding strategy $s_\v$ corresponding to thresholds $\v$, which sets $s_\v(v) \coloneqq b_k$ for $v \in \left(v_i, v_{i+1} \right]$.
	\end{itemize}
	This is because, conditioned on highest competing bids $\{h_s\}_{s=1}^t$, applying \eqref{eq:threshold_robust-inter-1} to iterate $\p_t$ of Algorithm~\ref{alg:threshold} and highest competing bid $h_t$ yields
	\begin{align*}
		\Phi(\v_{t+1}) - \Phi(\v_t) + \rev(\v_t, h_t) \leq \mye(F) + \eta \bar f\,.
	\end{align*}
    Next, noting that $\rev(\v_t, h_t) = \rev(A_t, h_t)$ and taking expectation over $\{h_s\}_{s=1}^t$ yields
    \begin{align*}
		\E[\Phi(\v_{t+1})] - \E[\Phi(\v_t)] + \E[\rev(\v_t, h_t)] \leq \mye(F) + \eta \bar f\,.
	\end{align*}
 
	Summing over all times steps, we get
	\begin{align*}
		&\sum_{t=1}^T \E[\Phi(\v_{t+1})] - \E[\Phi(\v_t)] + \sum_{t=1}^T \E[\rev(A_t, h_t)] \leq \mye(F) \cdot T + \eta \bar f \cdot T\\
		\implies &\sum_{t=1}^T \E[\rev(A_t, h_t)] \leq \mye(F) + \eta \bar f \cdot T + \Phi(\v_1) - \E[\Phi(\v_{T+1})] \leq  \mye(F) + \eta \bar f T + K/\eta\,,
	\end{align*}
	and the theorem statement follows. To complete the proof, we next establish \eqref{eq:threshold_robust-inter-1}.
	
	Fix some $\v \in \vv$ and highest competing bid $h = b_k$ for some $k \in \{0, \dots, K\}$. First, observe that
	\begin{align*}
		\rev(\v,h) = \rev(s_\v, h) &= \E_v\left[ s_\v(v) \cdot \mathbf 1 (s_\v(v) \geq b_k) \right]\\
		&= \sum_{i=k}^K b_i \cdot \mathbb P_v(s_\v(v) = b_i)\\
		&= \sum_{i=k}^K b_i \cdot \mathbb P_v(v \in (v_i, v_{i+1}])\\
		&= \sum_{i=k}^K b_i \cdot \mathbb (\{1- F(v_{i})\} - \{1- F(v_{i+1})\})\\
		&= b_k \cdot (1 - F(v_k)) +  \sum_{i=k+1}^K (1 - F(v_i)) \cdot (b_i - b_{i-1})\\
		&= b_k \cdot (1 - F(v_k)) + \sum_{i=k+1}^K (1 - F(v_i)) \cdot \epsilon \tag{$\clubsuit$}\label{eq:threshold_robust-inter-2}
	\end{align*}
	
	Next, using the terminology and the result from Corollary~\ref{cor:threshold_update}, we can write
	\begin{align*}
		\Delta \Phi(\v) = & \Phi(\v') - \Phi(\v)\\
		= & \frac{1}{\eta} \sum_{j=1}^K \int_{v_j'}^{v_j} \{1 - F(t)\} \cdot dt\\
		= & \frac{1}{\eta} \cdot \sum_{j=1}^{m-1} \int_{v_j}^{v_j} \{1 - F(t)\} \cdot dt + \frac{1}{\eta} \cdot \sum_{j=m}^k \int_{x}^{v_j} \{1 - F(t)\} \cdot dt\\
		 & + \frac{1}{\eta} \cdot \sum_{j=k+1}^{\ell-1} \int_{v_j + \eta \epsilon}^{v_j} \{1 - F(t)\} \cdot dt + \frac{1}{\eta} \cdot\sum_{j=\ell}^K \int_{1}^{v_j} \{1 - F(t)\} \cdot dt \\
		 = & \frac{1}{\eta} \cdot \sum_{j=m}^k \int_{x}^{v_j} \{1 - F(t)\} \cdot dt - \frac{1}{\eta} \cdot \sum_{j=k+1}^{\ell-1} \int_{v_j}^{v_j + \eta \epsilon} \{1 - F(t)\} \cdot dt - \frac{1}{\eta} \cdot\sum_{j=\ell}^K \int_{v_j}^{1} \{1 - F(t)\} \cdot dt \\
		 \leq & \frac{1}{\eta} \cdot \sum_{j=m}^k (v_j - x) \cdot (1 - F(x)) - \frac{1}{\eta} \cdot \sum_{j=k+1}^{\ell-1} (\eta \epsilon) \cdot (1 - F(v_j + \eta \epsilon))\\
		 \leq & \frac{1 - F(x)}{\eta} \cdot \sum_{j=m}^k (v_j - x) - \sum_{j=k+1}^{\ell-1} \epsilon \cdot (1 - F(v_j)) + \sum_{j=k+1}^{\ell-1} \epsilon \cdot (F(v_{j} + \eta \epsilon) - F(v_j))\\
		 \leq & \frac{1 - F(x)}{\eta} \cdot  \eta \cdot (v_k - b_k) - \sum_{j=k+1}^{\ell-1} (1 - F(v_j)) \cdot \epsilon + (\ell-k-2) \cdot \epsilon \cdot \bar f \eta \epsilon\\
		 = & (1 - F(v_k)) \cdot (v_k - b_k) + (F(v_k)-F(x)) \cdot (v_k - b_k) - \sum_{j=k+1}^{\ell-1} (1 - F(v_j)) \cdot \epsilon + (\ell-k-2)\eta \bar f \epsilon^2\\
		 \leq & (1 - F(v_k)) \cdot (v_k - b_k) + \eta \bar f \cdot (1- \epsilon) - \sum_{j=k+1}^{\ell-1} (1 - F(v_j)) \cdot \epsilon + (\ell-k-2)\eta \bar f \epsilon^2 \tag{$\spadesuit$}\label{eq:threshold_robust-inter-3}
	\end{align*}
	where the first inequality follows from the fact that $\int_{x}^{y} h(t) dt \leq (y-x) h(x)$ for a decreasing function $h:[0,1] \to [0,1]$ and $x \leq y$, the third inequality follows from the definition of $x$ (as defined in Lemma~\ref{lemma:GA_update}), and the final inequality follows from the fact that $v_k - x \leq \eta$ and $v_j - b_j \leq 1- \epsilon$ for all $j \in \{0,1,\dots, K\}$. Moreover, we have repeatedly used the fact that $F(y) - F(x) \leq \bar f \cdot (y-x)$ for all $0 \leq x \leq y \leq 1$.
	
	Combining \eqref{eq:strat_robust-inter-2} and \eqref{eq:strat_robust-inter-3} yields
	\begin{align*}
		\Delta \Phi(\v) + \rev(\v, h)
		\leq\ & (1 - F(v_k)) \cdot (v_k - b_k) + \eta \bar f \cdot (1- \epsilon) - \sum_{j=k+1}^{\ell-1} (1 - F(v_j)) \cdot \epsilon + (\ell-k-2)\eta \bar f \epsilon^2\\
		& + b_k \cdot (1 - F(v_k)) + \sum_{i=k+1}^K (1 - F(v_i)) \cdot \epsilon\\
		\leq\ & \mye(F) + \eta \bar f \cdot (1-\epsilon) + (\ell-k-2)\eta \bar f \epsilon^2 + \sum_{j=\ell}^K (1 - F(v_j)) \cdot \epsilon\\
		\leq\ & \mye(F) + \eta \bar f \cdot (1-\epsilon)+   (\ell-k-2)\eta \bar f \epsilon^2 + \sum_{j=\ell}^K \eta \bar f \epsilon^2\\
		\leq\ & \mye(F) + \eta \bar f \cdot (1-\epsilon)+ \eta \bar f \epsilon^2 \cdot K\\
		\leq\ & \mye(F) + \eta \bar f \cdot (1-\epsilon)+ \eta \bar f \cdot  \epsilon\\
		=\ & \mye(F) + \eta \bar f\,,
	\end{align*} 
	where the second inequality follows from $\mye(F) \geq v_k \cdot (1 - F(v_k))$, the third inequality follows from the fact that $v_j \geq 1 - \eta\epsilon$ for all $j \geq \ell$ (see Corollary~\ref{cor:threshold_update} for definition of $\ell$), and the fifth inequality follows from the assumption that $b_K = \epsilon K \leq 1$. Thus, we have established \eqref{eq:threshold_robust-inter-1} and thereby the theorem.
\end{proof}

\subsection{Proof of Theorem~\ref{thm:ic_threshold}}

\begin{proof}
	Fix a misreport map $M:[0,1] \to [0,1]$. Consider a value $v^* \in [0,1]$ and define the potential function $\Phi:\vv \to [-K/\eta,K/\eta]$ as follows:
	\begin{align*}
		\Phi(\v|M, v^*) \coloneqq \frac{1}{\eta} \cdot \left\{ \sum_{j=1}^K (v^* - v_j) \cdot \mathbf 1 (v^* > v_j) - \sum_{j=1}^{K} (M(v^*) - v_j) \cdot \mathbf{1}(M(v^*) > v_j) \right\}
	\end{align*}
	We start by showing that, to prove the theorem, it suffices to prove the following statement for all thresholds $\v \in \vv$, value $v^* \in [0,1]$, and highest competing bid $h$:
	\begin{align}\label{eq:ic_threshold_inter-1}
		\Delta \Phi(\v|M,  v^*) + R(\v | M, v^*, h) \leq 3 \cdot \mathbf{1}\left(\min_{j \in [K]} |v_j - v^*| \leq \eta \right)
	\end{align}
	where
	\begin{itemize}
		\item $\Delta \Phi(\v|M, v^*) \coloneqq \Phi(\v'|M, v^*) - \Phi(\v|M, v^*)$ is the change in potential caused by a single update step of Algorithm~\ref{alg:threshold}. Here $\v'$ are the thresholds obtained by applying the update step of Algorithm~\ref{alg:threshold} to $\v$ (see equation \eqref{eq:threshold_projection} of Corollary~\ref{cor:threshold_update} for a formal definition).
		\item $R(\v | M, v^*, h) \coloneqq \{v^* - s_\v(M(v^*))\} \cdot \mathbf{1}(h \leq s_\v(M(v^*)) - \{v^* - s_\v(v^*)\} \cdot \mathbf{1}(h \leq s_\v(v^*))$ is the regret associated with reporting the true value $v^*$ in lieu of misreporting $M(v^*)$. Here $s_\v$ is the bidding strategy corresponding to thresholds $\v$, i.e., $s_\v(v^*) = b_j$ if $v^* \in (v_j, v_{j+1}]$.
	\end{itemize}
	
	Suppose \eqref{eq:ic_threshold_inter-1} holds for all thresholds $\v \in \vv$, values $v^* \in [0,1]$, and highest competing bids $h$. Then, conditioned on $\{h_s\}_{s=1}^t$, we can apply it to $\v_t$ and $h_t$ to get
	\begin{align*}
		&\Phi(\v_{t+1}|M, V_t) - \Phi(\v_t|M, V_t) + R(\v_t|M, V_t,h_t) \leq 3 \cdot \mathbf{1}\left(\min_{j \in [K]} |v_{t,j} - V_t| \leq \eta \right) &&\forall t \in [T]\,. 
	\end{align*}
	Taking an expectation over $V_t \sim F$, and using $\E_{V_t}[\cdot]$ to denote the conditional expectation $\E[\cdot | \{h_s\}_{s=1}^t]$, yields
    \begin{align*}
        &\E_{V_t}[\Phi(\v_{t+1}|M, V_t)] - \E_{V_t}\Phi(\v_t|M, V_t)] + \E_{V_t}\left[ R(\v_t|M, V_t,h_t)  \right] \leq\ 3 \cdot \mathbb{P}_{V_t}\left(\min_{j \in [K]} |v_{t,j} - V_t| \leq \eta \right)\\
        \implies &\E_{V_t}[\Phi(\v_{t+1}|M, V_t)] - \E_{V_t}\Phi(\v_t|M, V_t)] + u(s_\v \circ M|F,h_t) - u(s_\v|F,h_t) \leq\ 3 \cdot \sum_{j=1}^K \mathbb{P}_{V_t}\left( |v_{t,j} -V_t| \leq \eta \right)\\
        \implies & u(s_\v \circ M|F,h_t) - u(s_\v|F,h_t) \leq 3 \cdot K \cdot \bar f \cdot 2\eta + \E_{V_t}[\Phi(\v_{t}|M, V_t)] - \E_{V_t}[\Phi(\v_{t+1}|M, V_t) ]\\
        \implies & u(s_\v \circ M|F,h_t) - u(s_\v|F,h_t) \leq 3 \cdot K \cdot \bar f \cdot 2\eta + \E_{V_t}[\Phi(\v_{t}|M, V_t)] - \E_{V_{t+1}}[\Phi(\v_{t+1}|M, V_{t+1}) ]\,,
    \end{align*}
    where the first implication follows from the definition of $u(s|F,h)$ and the union bound, the second implication follows from $\mathbb{P}\left(|v_{t,j} - v^*| \leq \eta \right) = \mathbb{P}\left(v^* \in  [v_{t,j} - \eta, v_{t,j} + \eta] \right) \leq \bar f \cdot 2\eta$ (here $\bar f$ is an upper bound on the density of $F$), and the third implication follows from the independence of $V_t$, $V_{t+1}$ and $\{h_s\}_{s=1}^t$.
    
	Taking expectation over $\{h_s\}_{s=1}^T$ and summing over $t \in [T]$ yields:
	\begin{align*}
		 \sum_{t=1}^T \E[u(s_\v \circ M|F,h_t)] - \sum_{t=1}^T \E[u(s_\v|F,h_t)] &\leq 6\bar f K \cdot \eta T + \sum_{t=1}^T \left\{\E[\Phi(\v_{t}|M, V_t)] - \E[\Phi(\v_{t+1}|M, V_{t+1}) ] \right\}\\
        &\leq 6\bar f K \cdot \eta T + \E[\Phi(\v_{1}|M, V_1)] - \E[\Phi(\v_{T+1}|M, V_{T+1}) ]\\
        &\leq 6\bar f K \cdot \eta T + \frac{2K}{\eta}\,,
	\end{align*} 
    where $V_{T+1}$ is a fresh sample from $F$ that is independent of all other random variables. Therefore, we have shown that \eqref{eq:ic_threshold_inter-1} is a sufficient condition for the theorem---we focus on establishing \eqref{eq:ic_threshold_inter-1} in the remainder.
	
	Fix a value $v^*$, highest competing bid $h$ and thresholds $\v \in \vv$. First assume that $\min_{j \in [K]} |v_j - v^*| > \eta$. In particular, this implies that if $v_j > v^*$ (respectively $v_j < v^*$), then $v_j'>v^*$ (respectively $v_j' < v^*$), i.e., the thresholds don't cross $v^*$ during the update $\v \to \v'$. Let $s_\v(M(v^*)) = b_w$ be the bid under strategy $s_\v \circ M$ for value $v^*$ (i.e., $w = \max\{j \mid v_j < M(v^*) \}$), let $s_\v(v^*) = b_u$ be the bid under strategy $s_\v$ for value $v^*$ (i.e., $u = \max\{j \mid v_j \leq v^* \}$), and let $h = b_i$ be the highest competing bid. Since the thresholds don't cross $v^*$ during the update $\v \to \v'$, we have
	\small
	\begin{align*}
		&\ \Delta \Phi(\v|M, v^*)\\
		 = &\ \frac{1}{\eta} \cdot \left\{ \sum_{j=1}^u (v^* - v'_j) - \sum_{j=1}^{K} (M(v^*) - v'_j) \cdot \mathbf{1}(M(v^*) > v'_j) \right\} - \frac{1}{\eta} \cdot \left\{ \sum_{j=1}^u (v^* - v_j) - \sum_{j=1}^{K} (M(v^*) - v_j) \cdot \mathbf{1}(M(v^*) > v_j) \right\}\\
		 = &\ \frac{1}{\eta} \cdot \left\{- \sum_{j=1}^u (v_j' - v_j) + \sum_{j=1}^K (M(v^*) -v_j) \cdot \mathbf{1}(M(v^*) > v_j) -  \sum_{j=1}^K (M(v^*) -v'_j) \cdot \mathbf{1}(M(v^*) > v'_j) \right\}\\
		 = &\ \frac{1}{\eta} \cdot \left\{- \sum_{j=1}^u (v_j' - v_j) + \sum_{j=1}^K (M(v^*) -v_j) \cdot \mathbf{1}(M(v^*) > v_j) -  \sum_{j=1}^K (M(v^*) -v'_j) \cdot \mathbf{1}(M(v^*) > v'_j) \right\}\\
		 = &\ \frac{1}{\eta} \cdot \left\{- \sum_{j=1}^u (v_j' - v_j) + \sum_{j=1}^K (v_j' -v_j) \cdot \mathbf{1}(M(v^*) > v_j) -  \sum_{j=1}^K (M(v^*) -v'_j) \cdot (\mathbf{1}(M(v^*) > v'_j) - \mathbf{1}(M(v^*) > v_j)) \right\}\\
		 \leq &\ \frac{1}{\eta} \cdot \left\{- \sum_{j=1}^u (v_j' - v_j) + \sum_{j=1}^w (v_j' -v_j) \right\}\,,
	\end{align*}
	\normalsize
	where the last inequality follows from the fact that $M(v^*) - v_j'$ and $\mathbf{1}(M(v^*) > v'_j) - \mathbf{1}(M(v^*) > v_j)$ have the same sign for all possible values of $v_j, v_j', M(v^*) \in [0,1]$.
	
	We establish \eqref{eq:ic_threshold_inter-1} by separately analyzing the following mutually exclusive and exhaustive cases. These cases (and their analysis) is fairly similar to the cases in the proof of Theorem \ref{thm:regret_threshold}. The main difference is in the case where $b_{w} \geq b_i > b_u$ (the misreporting benchmark wins the item, but the algorithm doesn't), which we split into two cases depending on whether or not $b_{i} > v^*$.
	\begin{enumerate}
    \item $b_u \geq b_w \geq b_i$: The utility obtained by misreporting with $M$ is $\{v^* - s_\v(M(v^*))\}\cdot \mathbf 1(h \leq s_\v(M(v^*))) = v^* - b_w$ and the utility under $s_\v$ is $\{v^* - s_\v(v^*)\}\cdot \mathbf 1(h \leq s_\v(v^*)) = v^* - b_u$. Therefore, the regret $R(\v|M,v^*,h) = b_u - b_w = (u-w) \cdot \epsilon$. On the other hand, since $v_j' = v_j + \eta \epsilon $ for all $w < j \leq u$, we get $\Delta \Phi(\v|M, v^*) \leq - (u - w) \epsilon$.

    \item $b_{w} > b_u \geq b_i$: The utility obtained by misreporting with $M$ is $\{v^* - s_\v(M(v^*))\}\cdot \mathbf 1(h \leq s_\v(M(v^*))) = v^* - b_w$ and the utility under $s_\v$ is $\{v^* - s_\v(v^*)\}\cdot \mathbf 1(h \leq s_\v(v^*)) = v^* - b_u$. Therefore, the regret $R(\v|s^*,v^*,h) = -(b_{w} - b_u)$. On the other hand, since $v_j' \leq v_j + \eta \epsilon$ for all $u < j \leq w$, we get $\Phi(\v|v^*) \leq (w - u)\epsilon = b_{w} - b_u$.

    \item $b_{i} > b_u \geq b_{w}$: The utility obtained by misreporting with $M$ is $\{v^* - s_\v(M(v^*))\}\cdot \mathbf 1(h \leq s_\v(M(v^*))) = 0$ and the utility under $s_\v$ is $\{v^* - s_\v(v^*)\}\cdot \mathbf 1(h \leq s_\v(v^*)) = 0$ because $h = b_i > b_u \geq b_{w}$. Hence $R(\v|M,v^*,h) = 0$. On the other hand, $\Delta \Phi(\v|M,v^*) = 0$ because $u < m$ and $v_j' = v_j$ for all $j < m$. Here, $u < m$ follows from the fact that thresholds cannot cross $v^*$ during the update $\v \to \v'$.
    
    \item $b_i > b_{w} > b_u$: The utility obtained by misreporting with $M$ is $\{v^* - s_\v(M(v^*))\}\cdot \mathbf 1(h \leq s_\v(M(v^*))) = 0$ and the utility under $s_\v$ is $\{v^* - s_\v(v^*)\}\cdot \mathbf 1(h \leq s_\v(v^*)) = 0$ because $h = b_i > b_{w} > b_u$. Hence we have $R(\v|M, v^*, h) = 0$. On the other hand, $\Delta \Phi(\v|M, v^*) \leq 0$ because $v_j' \leq v_j$ for all $j \leq i$.
  
		\item $b_u \geq b_{i} > b_w$: The utility obtained by misreporting with $M$ is $\{v^* - s_\v(M(v^*))\}\cdot \mathbf 1(h \leq s_\v(M(v^*))) = 0$ because $h = b_i > b_w$, and the utility under $s_\v$ is $\{v^* - s_\v(v^*)\}\cdot \mathbf 1(h \leq s_\v(v^*)) = v^* - b_u$. Therefore, the regret $R(\v|M,v^*,h) = - (v^* - b_u)$. On the other hand, $v_j' = x$ for $m \leq j \leq i$ and $v_j' = v_j + \eta \epsilon$ for $i < j \leq u$. The definition of $x$ implies $\sum_{j=w+1}^{i} (v_j - v_j') \leq \sum_{j=m}^i (v_j - x) \leq \eta (v_i - b_i)$. Moreover, the definition of $u$ implies $v^* \geq v_u \geq v_i$. Thus, we get $\Delta \Phi(\v|M,v^*) \leq (v_i - b_i) - (u-i)\epsilon \leq v^* - b_u$.

		\item $b_{w} \geq b_i > b_u$ and $b_i \leq v^*$: The utility obtained by misreporting with $M$ is $\{v^* - s_\v(M(v^*))\}\cdot \mathbf 1(h \leq s_\v(M(v^*))) = v^* - b_w$ and the utility under $s_\v$ is $\{v^* - s_\v(v^*)\}\cdot \mathbf 1(h \leq s_\v(v^*)) = 0$ because $h = b_i > b_u$. Therefore, the regret $R(\v|M,v^*,h) = v^* - b_{w}$. On the other hand, $v_j' \leq v_j + \eta \epsilon$ for all $j > i$ and $v_j' = x$ for all $m \leq j \leq i$. The definition of $x$, the assumption that $b_i \leq v^* < v_i$, and the assumption that the thresholds don't cross $v^*$ implies $\sum_{j=m}^i (v_j -x) = \eta (v_i - b_i)$. As a consequence, we get $\sum_{j=u+1}^i (v'_j - v_j) = \sum_{j=m}^i (x - v_j) = -\eta (v_i - b_i)$. Moreover, the definition of $u$ implies $v_i > v^* > v_u$. Thus, we have $\Delta \Phi(\v|M, v^*) \leq (w - i)\epsilon - (v_i - b_i) \leq -(v^* - b_{i^*})$.
		
		\item $b_{w} \geq b_i > b_u$ and $b_i > v^*$: The utility obtained by misreporting with $M$ is $\{v^* - s_\v(M(v^*))\}\cdot \mathbf 1(h \leq s_\v(M(v^*))) = v^* - b_w$ and the utility under $s_\v$ is $\{v^* - s_\v(v^*)\}\cdot \mathbf 1(h \leq s_\v(v^*)) = 0$ because $h = b_i > b_u$. Therefore, the regret $R(\v|M,v^*,h) = v^* - b_{w} = v^* - b_i + (b_i - b_w) \leq b_i - b_w$. On the other hand, $v_j' \leq v_j + \eta \epsilon$ for all $j > i$ and $v_j' = x \leq v_j$ for all $m \leq j \leq i$. Thus, we have $\Delta \Phi(\v|M, v^*) \leq (w - i)\epsilon = b_w - b_i$.

	\end{enumerate}
	In all of the seven cases, we have established the desired bound $\Delta \Phi(\v|M, v^*) + R(\v|M, v^*, h) \leq 0$. 
	
	To complete the proof of \eqref{eq:ic_threshold_inter-1}, we now relax our assumption that $\min_{j \in [K]} |v_j - v^*| > \eta$, and consider the setting where $\min_{j \in [K]} |v_j - v^*| \leq \eta$. First, observe that Corollary~\ref{cor:threshold_update} implies
	\begin{align*}
         &v'_j - v_j = 0 &\text{if } j < m\\
         & \sum_{j=m}^i \{v'_j - v_j\}   \geq -\eta \\
         &v'_j - v_j   \leq \eta \cdot \epsilon &\text{if } i < j\,,
    \end{align*}
    which in turn implies $\sum_{j=1}^K |v_j'- v_j| \leq 2\eta$. Next, observe that the following function is 1-Lipschitz:
    \begin{align*}
        v_j \mapsto -(v_j - v^*) \cdot \mathbf{1}(v^* > v_j) + (v_j - M(v^*)) \cdot \mathbf{1}(M(v^*) > v_j)\,. 
    \end{align*}
    As a consequence, we get $\Delta \Phi(\v|M, v^*) \leq 2$ for all $\v \in \vv$. On the other hand, $R(\v|M, v^*, h) \leq 1$ for all $\v, M, v^*, h$. Combining the two, we get the desired bound $\Delta \Phi(\v|M, v^*) + R(\v|s^*, v^*, h) \leq 3$ in the case when $\min_{j \in [K]} |v_j - v^*| \leq \eta$. This establishes \eqref{eq:ic_threshold_inter-1} and completes the proof.
\end{proof}

\subsection{Proof of Theorem~\ref{thm:multi_buyer_robust}}

\begin{proof}
	Let $\{\v(i)_t\}_{t=1}^T$ be the iterates of Algorithm~\ref{alg:threshold} when all of the buyers simultaneously use Algorithm~\ref{alg:threshold} with $\eta = 1/\sqrt{\bar f T}$ to bid, the seller sets (potentially random and adaptive) reserve prices $\{r_t\}_{t=1}^T$ and the ties are broken using (random) ranking-based rules $\{\sigma_t\}_{t=1}^T$. Note that $\v(i)_t$ is a random variable that depends on the realized values $\{V(i)_t\}_{i,t}$, (where $V(i)_t$ denotes the value of buyer $i$ in auction $t$), the reserve prices and the tie-breaking rankings. Let $\Pi = (\{\v(i)_t\}_{i,t}, \{r_t\}_t, \{\sigma_t\}_t)$ denote all of the tuples of random variables that determine the run of an algorithm. Consider the first-price auction in which buyer $i$ bids according to thresholds $\v(i) \in \vv$ (i.e., using the strategy $s_\v$), the seller sets the reserve $r$ and ties are broken using the ranking $\sigma$. For a value tuple $(V(1), \dots, V(n))$, define $a_i(V(1), \dots, V(n)| \{\v(i)\}_i, r, \sigma)$ to be 1 if buyer $i$ wins this auction and 0 otherwise. Moreover, let $h(i)_t$ be the effective highest competing bid faced by buyer $i$ in this auction, i.e., $a_i(V(1), \dots, V(n) | \{\v(i)\}_i, r, \sigma) = 1$ if and only if buyer $i$ bids greater than or equal to $h(i)_t$.
	
	Define the following direct revelation mechanism $\mathcal{M}$ for $n$ buyers and a single item:
	\begin{enumerate}
		\item Ask all $n$ buyers to report their values. Let $Z(i)$ denote the value reported by buyer $i$.
		\item Define the allocation function $X: [0,1]^n \to \Delta^n$ as follows:
			\begin{align*}
				X_i(Z(1), \dots, Z(n)) \coloneqq \frac{1}{T} \cdot \sum_{t=1}^T \E_{\{\v(i)_t\}_i, r_t, \sigma_t} [\ a_i\left(Z(1), \dots, Z(n)| \{\v(i)_t\}_i, r_t, \sigma_t\right)\ ]
			\end{align*}
			and the payment rule $P:[0,1]^n \to \R_+^n$ as
			\begin{align*}
				P_i(Z(1), \dots, Z(n)) \coloneqq \frac{1}{T} \cdot \sum_{t=1}^T \E_{\{\v(i)_t\}, r_t, \sigma_t} \left[\ a_i\left(Z(1), \dots, Z(n)| \{\v(i)_t\}_i, r_t, \sigma_t\right)\ \cdot s_{\v(i)_t}(Z(i))   \right]\,.
			\end{align*}
	\end{enumerate}
	
	Note that the definition of $P$ implies that $\rev(A, \{F_i\}_i) = T \cdot \E[P(Z(1), \dots, Z(n))]$ is the total expected revenue generated when all buyers simultaneously employ Algorithm~\ref{alg:threshold}.

	We will use $x_i$ and $p_i$ to denote the interim allocation rule of $\mathcal{M}$, i.e.,
	\begin{align*}
		&x_i(z) \coloneqq \E_{Z(j) \sim F_j} \left[X_i(Z(1), \dots, Z(i-1), z, Z(i+1), \dots Z(n)) \right]\\
		\text{and }\quad &p_i(z) \coloneqq \E_{Z(j) \sim F_j} \left[P_i(Z(1), \dots, Z(i-1), z, Z(i+1), \dots Z(n)) \right]\,.
	\end{align*}
	It is easy to see that $x_i$ is monotonic for all buyers $i \in [n]$ because the allocation function $a_i( \cdot | \{\v(i)\}_i, r, \sigma)$ is monotonic in the $i$-th component. Define $u_i(z,y) = z \cdot x_i(y) - p_i(y)$ to be the expected utility when buyer $i$ has value $z$ but reports $y$. For a function $Q:[0,1] \to [0,1]$, these definitions along with Fubini's Theorem imply
	\small
	\begin{align*}
		&\E_{Z(i) \sim F_i}[u_i(Z(i), Q(Z(i)))]\\
		=&\ \E_{\Pi, \{Z(j)\}_{j\neq i}} \left[ \frac{1}{T} \cdot \sum_{t=1}^T \E_{Z_i \sim F_i}\left[\{ Z(i) - s_{\v(i)_t}(Q(Z(i)))\} \cdot a_i\left(Z(1), \dots, Q(Z(i)), \dots Z(n)| \{\v(i)_t\}_i, r_t, \sigma_t\right)  \right] \right]\\
		=&\ \E_{\Pi, \{Z(j)\}_{j\neq i}} \left[ \frac{1}{T} \cdot \sum_{t=1}^T u(s_{\v(i)_t} \circ Q| F_i, h(i)_t) \right]
	\end{align*}
	\normalsize
	
	Define the optimal misreport function $M_i:[0,1] \to [0,1]$ for buyer $i$ as $M_i(z) \coloneqq \argmax_{y \in [0,1]} u_i(z,y)$. Moreover, define the regret for not misreporting value $z$ to be $y$ as
	\begin{align*}
		\delta_i(z) = u_i(z, M_i(z)) - u_i(z, z)\,.
	\end{align*}
	Now, note that Theorem~\ref{thm:ic_threshold} implies
	\begin{align}\label{eq:ic_inter}
		\E_{Z(i) \sim F_i}[\delta_i(Z(i))] &= \E_{Z(i) \sim F_i}[u_i(Z(i), M_i(Z(i))) - u_i(Z(i), Z(i))] \nonumber\\
		&= \E_{\Pi, \{Z(j)\}_{j\neq i}} \left[ \frac{1}{T} \cdot \sum_{t=1}^T \left\{u(s_{\v(i)_t} \circ M_i| F_i, h(i)_t) - u(s_{\v(i)_t}| F_i, h(i)_t) \right\} \right] \nonumber\\
		&\leq \E_{\Pi, \{Z(j)\}_{j\neq i}} \left[ \frac{1}{T} \cdot 8K\bar f^{\frac{1}{2}} \cdot \sqrt{T} \right] \nonumber\\
		&= \frac{8K\bar f^{\frac{1}{2}}}{\sqrt{T}}\,.
	\end{align}
	
	Next, observe that the definition of $\delta_i(\cdot)$ implies
	\begin{align*}
		u_i(z,z) + \delta_i(z) = u_i(z, M_i(z)) = \max_{y \in [0,1]} z \cdot x_i(y) - p_i(y) &&\forall z\in [0,1]\,.
	\end{align*}
	In other words, $u_i(z,z) + \delta_i(z)$ is the interim utility for value $z\in [0,1]$ in incentive compatible mechanism with allocation rule $X$ and payment rule $P$. Therefore Equation (5.7) of \citet{krishna2009auction} applies, and we get
	\begin{align*}
		u_i(z,z) + \delta_i(z) = \int_{0}^z x_i(y) \cdot dy - p_i(0) &&\forall\ z \in [0,1]\,.
	\end{align*}  
	Since $s_{\v}(0) = 0$ for all $\v \in \vv$, we have $p_i(0) = 0$, and as a consequence
	\begin{align*}
		p_i(z) = z \cdot x_i(z) - u_i(z,z) = \underbrace{z \cdot x_i(z)   - \int_{0}^z x_i(y) \cdot dy}_{\tilde p_i(z)} + \delta_i(z)
	\end{align*}
	Next, note that Revenue Equivalence (Proposition~5.2 of \citealt{krishna2009auction}) implies that $\tilde p_i(z)$ is simply the interim expected payment of buyer $i$ with value $z$ under the incentive-compatible mechanism with allocation rule $X$. Therefore, we have
	\begin{align*}
		\E_{\{Z(i)\}_i \sim \prod_i F_i} [P(Z(1), \dots, Z(n))] &= \sum_{i=1}^n \E_{Z(i) \sim F_i}[p_i(Z(i))]\\
		&= \sum_{i=1}^n \E_{Z(i) \sim F_i}[\tilde p_i(Z(i)) + \delta_i(Z(i))]\\
		&= \sum_{i=1}^n \E_{Z(i) \sim F_i}[\tilde p_i(Z(i))] + \sum_{i=1}^n \E_{Z(i) \sim F_i}[\delta_i(Z(i))]\\
		&\leq \mye(\{F_i\}_i) + n \cdot \frac{8K\bar f^{\frac{1}{2}}}{\sqrt{T}}
	\end{align*}
	where the last inequality follows from \eqref{eq:ic_inter}. Finally, the theorem follows from the fact that $\rev(A, \{F_i\}_i) = T \cdot \E[P(Z(1), \dots, Z(n))]$.
\end{proof}

\section{Proofs for Section~\ref{sec:log_regret}}\label{appendix:log_regret}

For completeness, we formally state the time-varying step-size variant of Algorithm~\ref{alg:known_GA} here:

\begin{algorithm}[H]
	\textbf{Input:} Value distribution $F$, initial iterate $\pmb p_1 \in \pp$, and step size schedule $\{\eta_t\}$.\\
	\For{$t=1$ to $T$}
	{ 
		Observe value $V_t \sim F$;\\
		Bid $A_t(V_t) = b_i$ if $V_t \in \left(F^-\left(1- p_{t,i} \right), F^-\left(1 - p_{t,i+1} \right) \right]$;\\
		Observe highest competing bid $h_t \sim \d$;\\
		Update $\pmb p_t$ with a Gradient Ascent step: 
		\begin{align*}
			\pmb p_{t+1} = \argmin_{\pmb p \in \pp} \| \pmb p - \pmb p_t^+\| \quad \text{where} \quad \pmb p_t^+ = \pmb p_t + \eta_t \cdot \grad u(\pmb p_t|F, h_t)
		\end{align*}
    }
   \caption{Gradient Ascent with Known Value Distribution and Time-Varying Step Sizes}
\end{algorithm}

\subsection{Proof of Proposition~\ref{prop:strong_concavity}}
\begin{proof}
Recall the definition of the utility function under highest-competing-bid distribution $\d$ as given in \eqref{eq:convex-refor}:
\begin{align*}
	u(\pmb p| F, \pmb d) &\coloneqq \sum_{i=0}^K d_i \cdot \left( \int_{1 - p_i}^1 F^- (u) \cdot du - \sum_{j=i}^K b_j \cdot (p_j - p_{j+1}) \right)\\
	&= \sum_{i=0}^K d_i \cdot \left( \int_{1 - p_i}^1 F^- (u) \cdot du - b_i \cdot p_i - \sum_{j=i+1}^K (b_j - b_{j-1}) \cdot p_j \right)\,.
\end{align*}

We start by showing the function $g:[0,1] \to \R_+$ defined below is $(1/\bar f)$-strongly concave:
\begin{align*}
	g(p) \coloneqq  \int_{1 - p}^1 F^- (u) \cdot du\,.
\end{align*}

To see this, observe that $g'(p) = F^-(1-p)$, and for $0 \leq \tilde p < p \leq 1$ we have
\begin{align*}
	g'(p) - g'(\tilde p) =  F^-(1-p) - F^-(1-\tilde p) \leq -\frac{1}{\bar f} \cdot(p - \tilde p)\,,
\end{align*}
where we have used the fact that $F(\tilde x) - F(x) \leq \bar f \cdot (\tilde x - x)$ for $0 \leq x < \tilde x \leq 1$. This allows us to establish the $(1/\bar f)$-strongly concavity of $g$ using the first-order condition:
\begin{align*}
	(g'(p) - g'(\tilde p))\cdot (p - \tilde p) \leq -\frac{1}{\bar f} \cdot (p - \tilde p)^2 \,.
\end{align*}

Next, observe that the function $G(\p)$ defined as
\begin{align*}
	G(\p) \coloneqq \sum_{i=0}^K d_i \cdot g(p_i)
\end{align*}
is $(d_\min/\bar f)$-strongly convex because
\begin{align*}
	(\nabla G(\p) - \nabla G(\tilde \p))^\top (\p - \tilde \p) = \sum_{i=1}^K d_i \cdot (g'(p_i) - g'(\tilde p_i))\cdot (p_i - \tilde p_i) \leq - \frac{d_\min}{\bar f} \cdot \|\p- \tilde \p\|^2\,.
\end{align*}
As $\p \mapsto u(\p|F,\d)$ is the sum of $G(\cdot)$ and a linear function of $\p$, we get the proposition.
\end{proof}

\subsection{Proof of Theorem~\ref{thm:log_regret}}

\begin{proof}
    First, observe that Algorithm~\ref{alg:known_GA} implements Stochastic Gradient Ascent for the reward function $\p \mapsto u(\p|F,\d)$:
\begin{align*}
    \nabla u(\p|F,\d) = \nabla \left(\sum_{i=0}^K d_i \cdot u(\p|F,b_i) \right) = \sum_{i=0}^K d_i \cdot \nabla u(\p|F,b_i) = \E_{h \sim \d} [\nabla u(\p|F,h)]\,.
\end{align*}
Moreover note that $\|\nabla u(\p|F,h)\| \leq 2 = G$ is an upper bound on the gradient samples.

Now, we can apply the $O(\log T)$-regret guarantee described in \citet{hazan2006logarithmic} for Stochastic Gradient Descent with step size $\eta_t = 1/(\alpha t)$ to get
\begin{align*}
		\max_{\p^*}\ \sum_{t=1}^T u(\p^*|F,\d) - \sum_{t=1}^T u(\p_t|F,\d) \leq \frac{G^2}{2\alpha} \cdot (1 + \log T) = \frac{2 \bar f}{d_\min} \cdot (1 + \log T)\,.
\end{align*}

Finally, Theorem~\ref{thm:convex-refor} implies that $u(A_t|F,\d) = u(\p_t|F,\d)$ for all $t \in [T]$ and $\max_{\p^*} u(\p^*|F,\d) = \max_{s^*} u(s^*|F,\d)$, which completes the proof.
\end{proof}